\documentclass{fundam}
\usepackage{amsfonts}

\usepackage{tabularx,lipsum,environ}

\usepackage{url} 
\usepackage[ruled,lined]{algorithm2e}
\usepackage{graphicx}
\usepackage{tikz}
\usetikzlibrary{automata, positioning, arrows}
\usetikzlibrary{graphs,quotes,petri, shapes}
\makeatletter
\newcommand{\problemtitle}[1]{\gdef\@problemtitle{#1}}
\newcommand{\probleminput}[1]{\gdef\@probleminput{#1}}
\newcommand{\problemquestion}[1]{\gdef\@problemquestion{#1}}
\NewEnviron{decisionproblem}{
  \problemtitle{}\probleminput{}\problemquestion{}
  \BODY
  \par\addvspace{.5\baselineskip}
  \noindent
  \begin{tabularx}{0.97\textwidth}{@{\hspace{\parindent}} l X c}
    \multicolumn{2}{@{\hspace{\parindent}}l}{\@problemtitle} \\
    \emph{Input:} & \@probleminput \\
    \emph{Question:} & \@problemquestion
  \end{tabularx}
  \par\addvspace{.5\baselineskip}
}
\makeatother

\newlength{\problemoffset}
\newcommand{\preset}[1]{\,^{\!\bullet}#1}
\newcommand{\postset}[1]{#1^{\bullet}}
\newcommand{\escale}[1]{\ensuremath{\scalebox{0.8}{#1}}}
\newcommand{\nscale}[1]{\ensuremath{\scalebox{0.8}{#1}}}

\newcommand{\myEdge}[2]{ \tikz[baseline=-3pt]{
\draw[#2,line width=0.3pt] (0,0) -- ++(0.6,0) node[anchor=base, yshift=2pt, pos=0.5] {\escale{$#1$}};
}}

\newcommand{\Edge}[1]{ \tikz[baseline=-1pt]{
\draw[->,line width=0.3pt] (0,0) -- ++(0.6,0) node[anchor=base, yshift=5pt, pos=0.5] {\escale{$#1$}};
}}

\newcommand{\drop}[1]{}
\newcommand{\edge}[1]{\myEdge{#1}{->}}

\newcommand{\ers}{\textsc{ERS}}
\newcommand{\hs}{\textsc{HS}}
\newcommand{\syn}{\textsc{Synthesis}}
\newcommand{\eff}{\textit{eff}}
\newcommand{\T}{\mathcal{T}}
\newcommand{\R}{\mathcal{R}}
\newcommand{\mS}{\ensuremath{\mathfrak{S}}}

\newcommand{\U}{\ensuremath{\mathfrak{U}}}
\tikzstyle{place}=[circle,thick,draw=blue!75,fill=blue!20,minimum size=6mm]
  \tikzstyle{red place}=[place,draw=red!75,fill=red!20]
  \tikzstyle{transition}=[rectangle,thick,draw=black!75,
  			  fill=black!20,minimum size=4mm]
			   \tikzstyle{every label}=[red]

\begin{document}

\setcounter{page}{139}
\publyear{22}
\papernumber{2135}
\volume{187}
\issue{2-4}

 \finalVersionForARXIV

\title{Synthesis of Pure and Impure Petri nets with \\ Restricted Place-environments: Complexity Issues}

\author{Raymond Devillers\\
D\'epartement d'Informatique - Universit\'e Libre de Bruxelles\\
 Boulevard du Triomphe,  B1050 Brussels, Belgium\\
 rdevil@ulb.ac.be
 \and Ronny Tredup\thanks{Address of correspondence:  Institut F\"ur Informatik - Universit\"at Rostock, Albert-Einstein-Stra{\ss}e 22,
                              D18059 Rostock, Germany}
 \\
 Institut F\"ur Informatik - Universit\"at Rostock\\
  Albert-Einstein-Stra{\ss}e 22,   D18059 Rostock, Germany\\
  ronny.tredup@uni-rostock.de
 }

\runninghead{R. Devillers and R. Tredup}{Synthesis with Restricted Place-environments}

\maketitle

\begin{abstract}
Petri net synthesis consists in deciding for a given transition system $A$ whether there exists a Petri net $N$ whose reachability graph is isomorphic to $A$.
Several works examined the synthesis of Petri net subclasses that restrict, for every place $p$ of the net, the cardinality of its preset or of its postset or both \emph{in advance} by \emph{small} natural numbers $\varrho$ and $\kappa$, respectively,
such as for example (weighted) marked graphs, (weighted) T-systems and choice-free nets.
In this paper, we study the synthesis aiming at Petri nets which have such restricted place environments, from the viewpoint of classical and parameterized complexity:
We first show that, for any fixed natural numbers $\varrho$ and $\kappa$, deciding whether for a given transition system $A$
there is a Petri net $N$ such that (1) its reachability graph is isomorphic to $A$ and (2) for every place $p$ of $N$ the preset of $p$
has at most $\varrho$ and the postset of $p$ has at most $\kappa$ elements is doable in polynomial time.
Secondly, we introduce a modified version of the problem, namely \textsc{Environment Restricted Synthesis} (\textsc{ERS},
for short), where $\varrho$ and $\kappa$ are part of the input, and show that \textsc{ERS} is NP-complete, regardless whether the sought net is impure or pure.
In case of the impure nets, our methods also imply that \textsc{ERS} parameterized by $\varrho+\kappa$ is $W[2]$-hard.
\end{abstract}

\section{Introduction}%

Petri net synthesis consists in deciding for a given transition system $A$ whether there is a Petri net $N$
such that the reachability graph of $N$ is isomorphic to $A$.
In the event of a positive decision, a possible solution $N$ should be constructed,
and in the event of a negative decision some reason(s) should be produced if possible.
Synthesis of Petri nets has applications in various fields:
For example, it is used to extract concurrency and distributability data from sequential specifications like transition systems
or languages \cite{fac/BadouelCD02};
it is applied in the field of process discovery to reconstruct a model from its execution traces \cite{daglib/0027363}
and in supervisory control for discrete event systems~\cite{deds/HollowayKG97};
and it is used for the synthesis of speed-independent circuits~\cite{tcad/CortadellaKKLY97}.

The synthesis problem has been originally solved for the class of
\emph{Elementary net systems}~\cite{acta/EhrenfeuchtR89}, relying on \emph{regions} of transition systems,
and has been found to be NP-complete for this class in~\cite{tcs/BadouelBD97}.
Later on, this solution was extended to weighted P/T nets, as well as to weighted pure ones,
for which, however, the synthesis problem is solvable in polynomial time~\cite{tapsoft/BadouelBD95}.

Since then, many studies have been carried out on the synthesis of structurally restricted subclasses of Petri nets,
which aim at improved (pre-) synthesis methods with regard to the specified subclass.
The most investigated subclasses of Petri nets include those that restrict the cardinality of the presets or the postsets
of the places by \emph{a priori} fixed \emph{small} natural numbers $\varrho$ and $\kappa$, respectively.
Among them are especially the so-called (weighted or plain)
\emph{marked graphs}~\cite{jcss/CommonerHEP71}
(every place has exactly one pre- and exactly one post-transition),
the (weighted) \emph{T-systems}~\cite{ac/Best86a} (every place has at most one pre- and at most one post-transition)
and, as a generalization of both, the (weighted) \emph{choice-free} nets~\cite{apn/TeruelCCS92,tsmc/TeruelCS97} (every place has at most one post-transition).
These restrictions are initially motivated by the fact that, from the theoretical point of view, the resulting net classes allow
a rich and elegant theory with respect to their structure as well as highly efficient analysis algorithms~
\cite{ac/Best86a,%
tsmc/TeruelCS97,
desel_esparza_1995,%
apn/HujsaDK14,
lata/BestD14,%
acta/BestDS18}.
From the perspective of practical applications, they are particularly useful in, for example, some applications like hardware
design~\cite{tcad/CortadellaKKLY97,%
707587}
or as a proper model for systems with bulk services and
arrivals~\cite{tsmc/TeruelCS97}.
On the other hand, as already mentioned, these classes have also been the subject of research aiming at Petri net synthesis
for many years~\cite{acta/BestDS18,%
acta/BestD09,%
fuin/BestD15,%
tcs/BestHW18,%
topnoc/DevillersEH19,%
fuin/DevillersH19}.
It turned out that these net classes provide some very useful features like, for example, persistency of their reachability
graphs~\cite{tsmc/TeruelCS97} that --if it comes to complexity issues-- 
allow improved synthesis procedures
that --instead of regions-- 
rather rely on some basic structural properties of the input transition system.
Also the computational complexity of synthesis depending on the desired subclass has been subject of this research:
In~\cite{acta/BestDS18}, for example, it has been shown that synthesis aiming at choice-free nets is polynomial,
and in~\cite{corr/abs-1910-14387}, for example, it has been proved that synthesis aiming at weighted marked
graphs (or weighted $T$-systems) is polynomial when the input transition system is circular.

In this paper, we extend the research on the computational complexity of synthesis aiming at Petri nets with restricted place environments:
We show that, for any fixed natural numbers $\varrho$ and $\kappa$, deciding whether for a given transition system $A$
there is a Petri net $N$ such that (1) its reachability graph is isomorphic to $A$ and (2) for every place $p$ of $N$ the preset
of $p$ has at most $\varrho$ and the postset of $p$ has at most $\kappa$ elements may be done in polynomial-time.
In a natural way, the question arises whether synthesis remains polynomial if the bounds $\varrho$ and $\kappa$
are not fixed in advance, but are part of the input.
In this paper, we answer this question negatively and show first for the class of (impure) Petri nets that the corresponding decision problem (\textsc{Environment Restricted Synthesis, ERS for short)} is NP-complete.
We obtain this result by methods that give also information about the parameterized complexity of \textsc{ERS}
parameterized by $\varrho+\kappa$ (aiming at impure nets):
The proof for the membership of \textsc{ERS} implies that its parameterized version belongs to the complexity class XP.
The NP-hardness of \textsc{ERS} results from a polynomial-time reduction of the well-known problem \textsc{Hitting Set}, which is also a valid parameterized reduction.
Since \textsc{Hitting Set} is $W[2]$-complete, this implies that \textsc{ERS} parameterized by $\varrho+\kappa$ is $W[2]$-hard.
Hence, $\varrho+\kappa$ is unsuitable for fixed parameter tractability  (FPT)-approaches as pioneered in~\cite{dagstuhl/DowneyF92,txcs/DowneyF13}.

This paper is an extended version of~\cite{apn/Tredup21a}.
The main new result is the proof that ERS is NP-complete even if we restrict the addressed net class to pure Petri nets.
We obtain this result by a reduction of the problem \textsc{Cubic Monotone 1 in 3 3Sat},
which cannot be derived from the reduction for the impure nets.

\medskip
\emph{Further Related Work.}
For net classes for which the (underlying) unrestricted synthesis problem is already NP-complete as, for example, it is the case for
$b$-bounded Petri nets~\cite{apn/Tredup19} or an overwhelming amount of Boolean nets~\cite{topnoc/Tredup21} the problem \ers\ (or its corresponding formulation) is also NP-complete.
This can easily be shown by a trivial reduction from the unrestricted to the restricted problem.
In~\cite{rp/Tredup19}, it has been shown that \ers, formulated for $b$-bounded Petri nets, is NP-complete even if $\kappa=1$.
Moreover, in~\cite{sofsem/Tredup20}, it has been argued that the corresponding problem, although being in XP,
is $W[1]$-hard for these nets, when $\varrho+\kappa$ is considered as a parameter.
In~\cite{tamc/TE20,corr/abs-2009-08871}, it has been shown that the parametrized complexity of
 (the Boolean formulation of) \ers\ is $W[1]$-hard or $W[2]$-hard for a lot of Boolean Petri nets.
However, neither of these results imply the ones provided by the current\linebreak paper.

This paper is organized as follows.
Section~\ref{sec:prelis} introduces necessary definitions and provides some examples.
After that, Section~\ref{sec:main_result} provides the announced complexity results, and Section~\ref{sec:new_content} provides some results on the synthesis of pure place-environment restricted P/T nets.
Finally, Section~\ref{sec:conclusion} briefly closes the paper.

\section{Preliminaries}\label{sec:prelis}%

In this section, we introduce relevant basic notions around Petri net synthesis and provide some examples.

\begin{definition}[Transition Systems]\label{def:transition_system}
A (deterministic, labeled) \emph{transition system} (TS, for short) $A=(S,E, \delta,\iota)$ is a directed labeled graph with
the set of nodes $S$ (called \emph{states}), the set of labels $E$ (called \emph{events}),
the partial \emph{transition function} $\delta: S\times E \longrightarrow S$ and the \emph{initial state} $\iota\in S$.
Event $e$ \emph{occurs} at state $s$, denoted by $s\edge{e}$, if $\delta(s,e)$ is defined.
By $s\edge{\neg e}$, we denote that $e$ does not occur at $s$.
We abridge $\delta(s,e)=s'$ by $s\edge{e}s'$ and call the latter an \emph{edge}.
By $s\edge{e}s'\in A$, we denote that the edge $s\edge{e}s'$ is present in $A$.
We say $A$ is \emph{loop-free} if $s\edge{e}s'\in A$ implies $s\not=s'$.
A  sequence $s_0\edge{e_1}s_1, s_1\edge{e_2}s_2,\dots, s_{n-1}\edge{e_n}s_n$ of edges is called a (directed labeled) \emph{path} (from $s_0$ to $s_n$ in $A$).
$A$ is called \emph{initialized} if, for every state $s\in S$, we have $s=\iota$ or there is a path from $\iota$ to $s$.
\end{definition}

If a TS $A$ is not explicitly defined, then we refer to its components by $S(A)$ (states), $E(A)$ (events), $\delta_A$ (function),
$\iota_A$ (initial state).
In this paper, we investigate whether a TS corresponds to the reachability graph of a Petri net.
Since the latter are always initialized, we assume that all TS are initialized without explicitly mentioning this each time.
We shall only consider finite TS, i.e., that $S$ and $E$ (hence also $\delta$) are finite.
Moreover, we consider TS $A$ and $B$ to be essentially the same when isomorphic:

\begin{definition}[Isomorphic TS]
Two TS $A=(S,E,\delta,\iota)$ and $B=(S',E,\delta',\iota')$ with the same set of events are \emph{isomorphic}, denoted by $A\cong B$, if there is a bijection $\varphi:S\rightarrow S'$ such that $\varphi(\iota)=\iota'$ and $s\edge{e}s'\in A$ if and only if $\varphi(s)\edge{e}\varphi(s')\in B$.
\end{definition}

Starting from a certain behavior that is defined by a transition system, we look for a machine that implements this behavior,
namely a Petri net (more exactly a weighted P/T net):

\begin{definition}[Petri Nets]\label{def:petri_nets}
A \emph{Petri net} $N=(P,T,f,M_0)$ consists of finite and disjoint sets of \emph{places} $P$ and \emph{transitions} $T$,
a (total) \emph{flow} $f: ((P \times T) \cup (T \times P)) \rightarrow \mathbb{N}$
(interpreted as denoting the number of \emph{tokens} absorbed and produced by $t$ on $p$)
and an \emph{initial marking} $M_0: P \rightarrow \mathbb{N}$ (more generally, a marking is any function $M:P\rightarrow\mathbb{N}$).\\
The \emph{preset} of a place $p$ is defined by $ \preset{p}=\{t \in T \mid f(t,p)>0\}$, thus comprising all the transitions that  produce tokens on $p$, while the \emph{postset} of $p$ is defined by $\postset{p} =\{t\in T \mid f(p,t)>0\}$, thus comprising all  the transitions that consume tokens from $p$.
Notice that $\preset{p} \cap \postset{p}$ is not necessarily empty, but if this is true for each place $p\in P$ the net is said \emph{pure}.
For $\varrho, \kappa\in \mathbb{N}$, we say $p$ is \emph{$(\varrho, \kappa)$-restricted} if $\vert \preset{p}\vert\leq \varrho$
and $\vert \postset{p}\vert \leq \kappa$; note that it is also possible to only constrain the size of $\preset{p}$ or $\postset{p}$,
by considering $(\varrho, |T|)$-restricted and $(|T|, \kappa)$-restricted nets.\\
A transition $t\in T$ can \emph{fire} or \emph{occur} in a marking $M:P\rightarrow \mathbb{N}$, denoted by $M\edge{t}$, if $M(p)\geq f(p,t) $ for all places $p\in P$.
The firing of $t$ in marking $M$ leads to the marking $M'(p)=M(p)-f(p,t)+f(t,p)$ for all $p\in P$, denoted by $M\edge{t}M'$.
This notation extends to sequences $w \in T^*$ and the \emph{reachability set} $RS(N)=\{M \mid \exists w\in T^*: M_0\edge{w}M \}$ contains all reachable markings of $N$.
The \emph{reachability graph} of $N$ is the TS $A_N=(RS(N), T,\delta, M_0)$, where for every reachable marking $M$ of $N$ and transition $t \in T$ with $M \edge{t} M'$ the transition function $\delta$ of $A_N$ is defined by $\delta(M,t) = M'$
($\delta(M,t)$ being undefined if $t$ may not fire in $M$).
\end{definition}

According to Definition~\ref{def:petri_nets}, for every Petri net, there is a TS, that reflects the global behavior of the net, namely its reachability graph.
However, not every TS is the behavior of a Petri net and thus the following decision problem arises:

\noindent
\fbox{\begin{minipage}[t][1.7\height][c]{0.97\textwidth}
\begin{decisionproblem}
  \problemtitle{\textsc{Synthesis}}
  \probleminput{A TS $A=(S,E,\delta, \iota)$.}
  \problemquestion{Does there exist a Petri net $N$ such that $A\cong A_N$?}
\end{decisionproblem}
\end{minipage}}
\bigskip

If \textsc{Synthesis} allows a positive decision, then we want to construct $N$ purely from $A$: $N$ is then called a
\emph{solution} of $A$. When there is a solution, there are infinitely many of them, sometimes with very different structure,
and it may be interesting to restrict the target class of a synthesis problem, like the $(\rho,\kappa)$-restricted class for instance.
Note that, since we only consider a finite TS $A$,
its solutions are always bounded, meaning that, for some integer $k$,
we have $M(p)\leq k$ for each place $p$ and reachable marking $M$.
Since $A$ and $A_N$ should be isomorphic, the events $E$ of $A$ become the transitions of $N$.
Note however that $A$ does not allow to answer all questions about its solutions: since we work up to isomorphism
it is not possible to state if some marking is reachable or dominated, but since $A$ is finite it will be easy to answer such questions when $N$ will be chosen.

The places, the flow and the initial marking of $N$, hence the solutions of $A$, originate from so-called \emph{regions}
of the TS $A$.

\begin{definition}[Region]\label{def:region}
A \emph{region} $R=(sup, con, pro)$ of a TS $A=(S, E, \delta, \iota)$ consists of three mappings
\emph{\underline{sup}port} $sup:S \rightarrow \mathbb{N}$,  \emph{\underline{con}sume} $con:E \rightarrow \mathbb{N}$
and \emph{\underline{pro}duce} $pro:E \rightarrow \mathbb{N}$,
such that if $s \edge{e} s'$ is an edge of $A$, then $con(e)\leq sup(s)$ and $sup(s')=sup(s)-con(e)+pro(e)$.
The \emph{preset} of $R$ is defined by $\preset{R}=\{e\in E \mid pro(e) > 0\}$
and its \emph{postset} by $\postset{R} =\{e\in E \mid con(e) > 0\}$.
A region is \emph{pure} if $\preset{R}\cap\postset{R}=\emptyset$.
For $\varrho,\kappa\in \mathbb{N}$, we shall say that $R$ is \emph{$(\varrho,\kappa)$-restricted} if
$\vert \preset{R}\vert\leq\varrho$ and $\vert \postset{R} \vert \leq \kappa $.
\end{definition}

\begin{remark}\label{rem:implicitly}
Notice that \emph{if} $R=(sup, con, pro)$ \emph{is} a region of a TS $A=(S,E,\delta,\iota)$, \emph{then}
$R$ can be obtained from $sup(\iota)$, $con$, and $pro$:
Since $A$ is initialized, for every state $s\in S$, there is a path $\iota\edge{e_1}\dots \edge{e_n}s_n$ such that $s=s_n$.
Hence, we inductively obtain $sup(s_{i+1})$ by $sup(s_{i+1})=sup(s_i) -con(e_{i+1}) +pro(e_{i+1})$
for all $i\in \{0,\dots, n-1\}$ and $s_0 = \iota$.
 For brevity, we shall often use this observation and present regions only \emph{implicitly} by $sup(\iota)$, $con$ and $pro$.
When a region will be defined that way, it will be necessary
to check that such a definition is coherent, i.e., each computed support is non-negative
and if two edges lead to the same state this also leads to the same support; such checks will however be usually easy.
For an even more compact presentation, for $c,p\in\mathbb{N}$, we shall group events with the same ``behavior'' together
by $\T_{c,p}^R=\{e\in E\mid con(e)=c\text{ and } pro(e)=p\}$.
\end{remark}

Regions of the TS become places in a sought net if it exists:
for a place $R=(sup, con, pro)$ of such a net,
$con(e)$ defines $f(R,e)$, the number of tokens that $e$ consumes from $R$, and $pro(e)$ defines $f(e,R)$, the number of tokens that $e$ produces on $R$, and $sup(s)$ models (the number of tokens) $M(R)$ (that are on $R$) in the marking $\varphi(s)=M$ of $N$ that corresponds to state $s$ of $A$ via the isomorphism $\varphi$ between $A$ and $A_N$.

\begin{definition}[Synthesized Net]\label{def:synthesized_net}
Every set $\mathcal{R} $ of regions of a TS $A$ defines the \emph{synthesized net} $N^{\mathcal{R}}_A=(\mathcal{R}, E, f, M_0)$ with $f(R,e)=con(e)$, $f(e,R)=pro(e)$ and $M_0(R)=sup(\iota)$ for all $R=(sup, con , pro)\in \mathcal{R}$ and all $e\in E$.
\end{definition}

In order to ensure that the input behavior $A$ is correctly captured by a synthesized net $N$,
meaning that $A$ and $A_N$ are identified by an isomorphism $\varphi$,
we have to ensure that distinct states $s\not=s'$ of $A$ correspond to distinct markings $\varphi(s)\not=\varphi(s')$ of $N$.
In particular, we need $A$ to have the \emph{state separation property}, which means that its \emph{state separation atoms} are solvable:
\begin{definition}[State Separation]\label{def:ssp}
A pair $(s, s')$ (unordered) of distinct states of $A$ defines a \emph{state separation atom} (SSA).
There are thus $|S|\cdot(|S|-1)/2$ such SSA's.
A region $R=(sup, con, pro)$ \emph{solves} $(s,s')$ if $sup(s)\not=sup(s')$.
We say a state $s$ is \emph{solvable} if, for every $s'\in S\setminus \{s\}$, there is a region that solves the SSA $(s,s')$.
If every SSA or, equivalently, every state of $A$ is solvable then $A$ has the \emph{state separation property} (SSP).
\end{definition}

Furthermore, we have to prevent the firing of a transition in a marking $M$, if its corresponding event does not occur at the state $s$ of $A$ that corresponds to $M$ via the isomorphism $\varphi$, that is, if $s\edge{\neg e}$, then $\varphi(s)\edge{\neg e}$, where $\varphi(s)=M$.
In particular, $A$ must have the \emph{event/state separation property}, meaning that all \emph{event/state separation atoms} of $A$ are solvable:
\begin{definition}[Event/State Separation]\label{def:essp}
A pair $(e,s)$ of event $e\in E $ and state $s\in S$ such that $ s\edge{\neg e}$ defines
an \emph{event/state separation atom} (ESSA).
There are thus $|S|\cdot|E|-|\delta|$ such ESSA's
(where $|\delta|$ is the number of pairs $(s,e)$ such that $\delta(s,e)$ is defined).
A region $R=(sup, con, pro)$ \emph{solves} $(e,s)$ if $sup(s) < con(e) $.
We say an event $e$ is \emph{solvable} if, for all $s\in S$ such that $s\edge{\neg e}$, there is a region that solves the ESSA $(e,s)$.
If every ESSA or, equivalently, every event of $A$ is solvable then $A$ has the \emph{event state separation property} (ESSP).
\end{definition}

\begin{definition}[Admissible set]\label{def:admissible_set}
A set $\mathcal{R}$ of regions of $A$ is called \emph{admissible} if it witnesses the SSP and the ESSP of $A$, that is, every SSA and ESSA of $A$ is solvable by a region $R$ of $\mathcal{R}$.
\end{definition}

The next lemma, borrowed from~\cite[p.~162]{txtcs/BadouelBD15} but also present in~\cite{DR-synth-96},
 establishes the connection between the existence of an admissible set $\mathcal{R}$ of $A$
 and the existence of a Petri net $N$ whose rechability graph is isomorphic to $A$.
Notice that Petri nets correspond to the type of nets $\tau_{PT}$ in~\cite[p.~130]{txtcs/BadouelBD15}.
\begin{lemma}[\cite{txtcs/BadouelBD15}]\label{lem:admissible}
If $A$ is a TS and $N$ a Petri net, then $A\cong A_N$ if and only if there is an admissible set $\mathcal{R}$ of $A$ and $N=N^{\mathcal{R}}_A$.
\end{lemma}
By Lemma~\ref{lem:admissible}, deciding the existence of a sought net $N$ for $A$ is equivalent to deciding
the existence of an admissible set  $\mathcal{R}$ of $A$.
Moreover, since the regions $R=(sup, con, pro)$ of $\mathcal{R}$ yield the places in $N=N_A^{\mathcal{R}}$ and the
corresponding flow is defined by $con$ and $pro$, the places of $N$ are $(\varrho,\kappa)$-restricted if and only if
every region $R\in \mathcal{R}$ is $(\varrho,\kappa)$-restricted.
Eventually, this leads us to the following decision problem, which is the main subject of this paper:

\noindent
\fbox{\begin{minipage}[t][1.7\height][c]{0.97\textwidth}
\begin{decisionproblem}
  \problemtitle{\textsc{Environment Restricted Synthesis}}
  \probleminput{A TS $A=(S,E,\delta, \iota)$ and two natural numbers $\varrho$ and $\kappa$.}
  \problemquestion{Does  there exist an admissible set $\mathcal{R}$ of $A$ such that every region $R\in \mathcal{R}$ satisfies $\vert \preset{R}\vert\leq \varrho$ and $\vert \postset{R}\vert \leq \kappa$?}
\end{decisionproblem}
\end{minipage}}

\begin{figure}[b!]
\vspace*{-4mm}
\begin{center}
\begin{minipage}{\textwidth}
\begin{center}
\begin{tikzpicture}[new set = import nodes]
\begin{scope}[nodes={set=import nodes}]
		\node (T) at (-1.5,0) {$A_1:$};
		\foreach \i in {0,1,2} { \coordinate (\i) at (\i*1.5cm, 0) ;}
		\foreach \i in {0} { \node (\i) at (\i) {\nscale{$\boldsymbol{s_\i}$}};}
		\foreach \i in {1,2} { \node (\i) at (\i) {\nscale{$s_\i$}};}
		\path (0) edge [->, out=-120,in=120,looseness=5] node[left] {\escale{$a$}} (0);
\graph {(import nodes);
			0 ->["\escale{$b$}"]1->["\escale{$a$}"]2;
		};
\end{scope}
\begin{scope}[xshift=6cm, nodes={set=import nodes}]
		\node (T) at (-1.5,0) {$A_2:$};
		\foreach \i in {0,1} { \coordinate (\i) at (\i*2cm, 0) ;}
		\foreach \i in {0} { \node (\i) at (\i) {\nscale{$\boldsymbol{s_\i}$}};}
		\foreach \i in {1} { \node (\i) at (\i) {\nscale{$s_\i$}};}
\graph {(import nodes);
			0 ->[bend left =15, "\escale{$a$}"]1;
			1 ->[bend left =15, "\escale{$a$}"]0;
		};
\end{scope}
\begin{scope}[xshift=11cm, yshift=-0.5cm,nodes={set=import nodes}]
		\node (T) at (-1.5,0.5) {$A_3:$};
		\node (0) at (0,0) {\nscale{$\boldsymbol{s_0}$}};
		\node (1) at (1.5,0) {\nscale{$s_1$}};
		\node (2) at (0,1) {\nscale{$s_2$}};
		\node (3) at (1.5,1) {\nscale{$s_3$}};
\graph {(import nodes);
			0 ->[swap,"\escale{$a$}"]1->[swap, "\escale{$b$}"]3;
			0->[ "\escale{$b$}"]2->["\escale{$a$}"]3;
		};
\end{scope}
\end{tikzpicture}
\end{center}\vspace*{-8mm}
\caption{The TS $A_1$ (Example~\ref{ex:nothing}), $A_2$ (Example~\ref{ex:essp}) and $A_3$ (Example~\ref{ex:ssp_and_essp}) (initial states are indicated in bold). 
}\label{fig:three_ts}
\end{minipage}
\begin{minipage}{\textwidth}
\begin{center}
\begin{tikzpicture}[new set = import nodes]

\begin{scope}[node distance=1.5cm,bend angle=45,auto]
	\node (T) at (-1.75,-0.5) {$N_1:$};
\node [place] (R1)[tokens =1,label=left: $R_1$, node distance =1.5cm]{};
\node [transition] (a)[right of =R1]{$a$}
	 edge [pre, above] node {$\nscale{1}$}  (R1);
\node [place] (R2)[tokens =1,right of =a, label=right: $R_3$, node distance =1.5cm]{}
	edge [pre, bend right=25, above] node {$\nscale{1}$}  (a)
	edge [post, bend left=25, below] node {$\nscale{1}$}  (a);

\node [place] (R3)[tokens =1,below of=R1, label=left: $R_2$, node distance =1cm]{};
\node [transition] (b)[right of =R3]{$b$}
	 edge [pre, above] node {$\nscale{1}$}  (R3);

\end{scope}
\begin{scope}[xshift=7cm, yshift=-1cm,nodes={set=import nodes}]
		\node (T) at (-1.5,0.5) {$A_{N_1}:$};
		\node (0) at (0,0) {$\boldsymbol{111}$};
		\node (1) at (2,0) {$011$};
		\node (2) at (0,1) {$101$};
		\node (3) at (2,1) {$001$};
\graph {(import nodes);
			0 ->[swap, "\escale{$a$}"]1->[swap, "\escale{$b$}"]3;
			0->["\escale{$b$}"]2->["\escale{$a$}"]3;
		};
\end{scope}

\end{tikzpicture}
\end{center}\vspace*{-4mm}
\caption{Left: The Petri Net $N_1$ with initial marking $M_0(R_1)M_0(R_2)M_0(R_3)=111$ and $(1,1)$-restricted places.
Right: The reachability graph $A_{N_1}$.}\label{fig:11_restricted_net}\vspace*{2mm}
\end{minipage}
\begin{minipage}{\textwidth}
\begin{center}
\begin{tikzpicture}[new set = import nodes]

\begin{scope}[node distance=1.5cm,bend angle=45,auto]
\node (T) at (-1.75,-0.5) {$N_2:$};
\node [place] (R1)[tokens =1,label=left: $R_1$, node distance =1.5cm]{};
\node [transition] (a)[right of =R1]{$a$}
	 edge [pre, above] node {$\nscale{1}$}  (R1);

\node [place] (R3)[tokens =1,below of=R1, label=left: $R_2$, node distance =1cm]{};
\node [transition] (b)[right of =R3]{$b$}
	 edge [pre, above] node {$\nscale{1}$}  (R3);

\end{scope}
\begin{scope}[xshift=6cm, yshift=-1cm,nodes={set=import nodes}]
		\node (T) at (-1.5,0.5) {$A_{N_2}:$};
		\node (0) at (0,0) {$\boldsymbol{11}$};
		\node (1) at (2,0) {$01$};
		\node (2) at (0,1) {$10$};
		\node (3) at (2,1) {$00$};
\graph {(import nodes);
			0 ->[swap, "\escale{$a$}"]1->[swap, "\escale{$b$}"]3;
			0->["\escale{$b$}"]2->["\escale{$a$}"]3;
		};
\end{scope}

\end{tikzpicture}
\end{center}\vspace*{-4mm}
\caption{Left: The Petri Net $N_2$ with initial marking $M_0(R_1)M_0(R_2)=11$ and $(0,1)$-restricted places.
Right: The reachability graph $A_{N_2}$. }\label{fig:01_restricted_net}
\end{minipage}
\end{center}
\end{figure}

\begin{example}\label{ex:nothing}
The TS $A_1$ of Figure~\ref{fig:three_ts} has neither the SSP nor the ESSP:
If $R=(sup, con, pro)$ is a region of $A_1$, then the edge $s_0\edge{a}s_0$ requires $sup(s_0)=sup(s_0)-con(a)+pro(a)$, implying $con(a)=pro(a)$.
The latter implies $sup(s_1)=sup(s_2)$ by $sup(s_2)=sup(s_1)-con(a)+pro(a)$.
Moreover, by $s_1\edge{a}$, we have $sup(s_1)\geq con(a)$ and thus $sup(s_2)\geq con(a)$.
Consequently, since $R$ was arbitrary, the SSA $(s_1,s_2)$ and the ESSA $(a, s_2)$ are not solvable.
\end{example}

\begin{example}\label{ex:essp}
The TS $A_2$ of Figure~\ref{fig:three_ts} has the ESSP by triviality, since the only event $a$ occurs at all states of $A_2$, but not the SSP:
The SSA $(s_0,s_1)$ is not solvable, since any region $R=(sup, con, pro)$ of $A_2$ satisfies $sup(s_0)=sup(s_1)-con(a)+pro(a)$ and $sup(s_1)=sup(s_0)-con(a)+pro(a)$, which implies $sup(s_0)=sup(s_1)$.
\end{example}

\begin{example}\label{ex:ssp_and_essp}
The TS $A_3$ of Figure~\ref{fig:three_ts} has the ESSP and the SSP:
The region $R_1=(sup_1,con_1,pro_1)$, which, according to Remark~\ref{rem:implicitly}, is implicitly given by $sup_1(s_0)=1$ and $\T_{1,0}^{R_1}=\{a\}$ and $\T_{0,0}^{R_1}=\{b\}$, solves $(a,s_1)$, $(a,s_3)$, $(s_0,s_1)$, $(s_0,s_3)$, $(s_2,s_1)$ and $(s_2,s_3)$.
We obtain $R_1$ explicitly by  $sup_1(s_1)=sup_1(s_0)-con_1(a)+pro_1(a)=0$ and $sup_1(s_2)=sup(s_0)-con_1(b)+pro_1(b)=1$ and $sup_1(s_3)=sup_1(s_2)-con_1(a)+pro_1(a)=0$.

\medskip
Moreover, the region $R_2=(sup_2,con_2,pro_2)$, which is defined by $sup_2(s_0)=sup_2(s_1)=1$, $sup_2(s_2)=sup_2(s_3)=0$ and $\T_{1,0}^{R_2}=\{b\}$ and $\T_{0,0}^{R_2}=\{a\}$, solves the remaining SSA $(s_0,s_2)$ and $(s_1,s_3)$ and ESSA $(b,s_2)$ and $(b,s_3)$ of $A$.

The TS $A_3$ has also the region $R_3=(sup_3, con_3, pro_3)$ defined by $sup_3(s_i)=1$ for all $i\in \{0,1,2,3\}$ and $\T_{1,1}^{R_3}=\{a\}$ and $\T_{0,0}^{R_3}=\{b\}$.

Since $R_1$ and $R_2$ solve all SSA and ESSA both of $\mathcal{R}_1=\{R_1,R_2, R_3\}$ and $\mathcal{R}_2=\{R_1,R_2\}$ are admissible sets of $A_3$.
Figure~\ref{fig:11_restricted_net} shows the synthesized net $N_1=N^{\mathcal{R}_1}_{A_3}$ whose places are $(1,1)$-restricted, but not $(0,1)$-restricted, since $\vert \preset{R_3}\vert =\vert \{a\}\vert =1$.
The reachability graph $A_{N_1}$ is sketched on the right hand side of Figure~\ref{fig:11_restricted_net} and it is isomorphic to $A_3$.
The isomorphism $\varphi$ is given by $\varphi(s_0)=111$, $\varphi(s_1)=011$, $\varphi(s_2)=101$ and $\varphi(s_3)=001$.
However, the input $(A_3,0,1)$ for \textsc{Environment Restricted Synthesis} allows a positive decision, because the admissible set $\mathcal{R}_2$ satisfies $\vert \preset{R}\vert =0$ and $\vert \postset{R}\vert =1$ for all $R\in \mathcal{R}_2$.
Figure~\ref{fig:01_restricted_net} shows the synthesized net $N_2=N^{\mathcal{R}_2}_{A_3}$ (left) and its reachability graph $A_{N_2}$ (right).
\end{example}

\section{The computational complexity of \textsc{Environment Restricted \\ Synthesis}}\label{sec:main_result}%

The following theorem provides the main contribution of this section:

\begin{theorem}\label{the:main_result}
Environment Restricted Synthesis is NP-complete.
\end{theorem}

In order to prove Theorem~\ref{the:main_result}, we have to show that \ers\ is in NP and that it is NP-hard.
The following Section~\ref{sec:in_np} is dedicated to the membership in NP.
The corresponding proof basically extends the deterministic polynomial-time algorithm for the (unrestricted) \textsc{Synthesis} to a non-deterministic algorithm for \ers.
The applied methods also show that, for any fixed $\varrho$ and $\kappa$, environment restricted synthesis can be done in polynomial-time.
After that, Section~\ref{sec:np_hard} deals with the hardness part and provides a polynomial reduction of
the \textsc{Hitting Set}  problem.

\subsection{Environment Restricted Synthesis is in NP}\label{sec:in_np}

In this section, we show that \ers\ belongs to the complexity class NP.
In order to obtain a fitting non-deterministic algorithm for \ers, we extend the deterministic approach for the (unrestricted) \syn, which, for example, has been presented in~\cite{txtcs/BadouelBD15}.
The tractability of \syn\  is based on the fact that the solvability of a single (state or event/state) separation atom $\alpha$ of a given TS $A=(S,E,\delta,\iota)$ is polynomial-time reducible to a system $L_\alpha$ composed of a polynomial number of linear equations and inequalities on rational variables.
Such a system can be solved or decided unsolvable in polynomial-time by Khachiyan's method and theorem~\cite[pp. 168-170]{networks/Rajan90}.
Since $A$ has at most $\vert E\vert\cdot \vert S\vert +\vert S\vert^2$ separation atoms, it follows that deciding the solvability of $A$ by deciding the solvability of every atom via the solvability of its corresponding systems $L_\alpha$, yields a deterministic polynomial-time algorithm for \syn\ (as well as a method to produce a solution when the decision is positive).
For a separation atom $\alpha$, our approach extends the aforementioned system $L_\alpha$ non-deterministically to
a system $L_\alpha'$.
The additional linear inequalities encode the restriction requirements:
$L_\alpha'$ is solvable if and only if there is a properly restricted region that solves $\alpha$.
The size of $L_\alpha'$ is still polynomially bounded by the size of $A$ and can again be solved with Khachiyan's method.
Hence, the membership of \ers\ in NP follows.

In the following, unless explicitly stated otherwise, let $(A,\varrho,\kappa)$ be an arbitrary but fixed input of \ers\ with TS $A=(S,\{e_1,\dots, e_n\},\delta,\iota)$, and let $\alpha$ be an arbitrary but fixed separation atom of $A$.
In~order to develop the announced approach, we first$\,$ briefly recapitulate$\,$ the deterministic approach~~to

\eject

\noindent
  solve $\alpha$ by a region, which is not necessarily restricted.
While we are only informally providing the intended functionality of all equations and inequalities presented, the formal proofs for the corresponding statements can be found in~\cite{txtcs/BadouelBD15}.
After that, we introduce the announced extension of this approach.

\medskip
First, by Remark~\ref{rem:implicitly}, a region $R=(sup, con, pro)$ is completely defined by $sup(\iota)$, $con$ and $pro$,
 i.e., by $2\cdot n+1$ variables. Hence, $R$ can be identified with a vector
 \[\textbf{x}=(sup(\iota),con(e_1),\dots, con(e_n),pro(e_1),\dots, pro(e_n))\in \mathbb{N}^{2n+1}\]

The aforementioned system $L_\alpha$ essentially consists of two parts:
The first part consists of equations and inequalities that ensure that a solution can be interpreted as a region at all.
Among others, this part encompasses equations that result from \emph{fundamental cycles}, which are defined by \emph{chords} of a \emph{spanning tree} of $A$.
The second part consists of an inequality that ensures that such a region actually solves $\alpha$.

A spanning tree of a TS is a sub-TS whose underlying undirected and unlabeled graph is a tree in the common graph-theoretical sense that is rooted at $\iota$:

\begin{definition}[Spanning tree, chord]\label{def:spanning_tree}
A \emph{spanning tree} $A'=(S,E',\delta',\iota)$ of the TS $A$ is a loop-free TS such that, for all $s,s'\in S$ and all $e\in E'$ the following is satisfied:
(1) if $s\edge{e}s'\in A'$, then $s\edge{e}s'\in A$;
(2) either $s=\iota$ or there is exactly one directed labeled path $P_s$ from $\iota$ to $s$ in $A'$.
An edge $s\edge{e}s'$ that is present in $A$ but not in $A'$ is called a \emph{chord} (for $A'$).
\end{definition}

In the following, let $A'$ be an arbitrary but fixed spanning tree of $A$.
By Definition~\ref{def:spanning_tree}, $A'$ contains $|\delta'|=(|S|-1)$ edges, and $(|\delta|-|S|+1)$ chords.
Moreover, for every state $s\in S\setminus\{\iota\}$, there is exactly one path $P_s$ from $\iota$ to $s$ in $A'$.
In order to count, for every event $e\in E$, the number of its occurrences along $P_s$, we use a \emph{Parikh-vector}:

\begin{definition}[Parikh-vector]\label{def:parikh-vector}
Let $s\in S$, and let $P_s=\iota\edge{e_{i_1}}\dots\edge{e_{i_m}}s$ be the unique path from $\iota$ to $s$ in $A'$ if $s\not=\iota$.
The \emph{Parikh-vector} $\psi_s$ (of $s$ in $A'$) is the mapping $\psi_s:\{e_1,\dots, e_n\}\rightarrow \mathbb{N}$ defined, for every $e\in \{e_1,\dots, e_n\}$, by
\[\psi_s(e)=
\begin{cases}
\vert \{\ell \mid e_{i_\ell}=e \text{ and } \ell\in \{1,\dots, m\}\}\vert, & \text{if } s\not=\iota\\
0, & \text{ otherwise}
\end{cases}
\]
For convenience we denote $\psi_s=(\psi_s(e_1),\dots, \psi_s(e_n))$.
\end{definition}

The chords for $A'$ define so-called \emph{fundamental cycles}:

\begin{definition}
Let $t=s\edge{e}s'$ be a chord for $A'$.
The \emph{fundamental cycle} $\psi_t$ (of $t$) is the mapping $\psi_t:\{e_1,\dots, e_n\}\rightarrow \mathbb{Z}$ defined, for all $i\in  \{1,\dots, n\}$, by
\[\psi_t(e_i)=
\begin{cases}
\psi_s(e_i) -\psi_{s'}(e_i), &\text{if } e_i\not=e\\
\psi_s(e) +1 - \psi_{s'}(e), &\text{otherwise}
\end{cases}
\]
For convenience, we denote $\psi_t=(\psi_t(e_1),\dots, \psi_t(e_n))$.
\end{definition}

\eject

A region $R=(sup, con,pro)$ of $A$ is completely defined from $sup(\iota)$ and $con,pro$ as follows:\\
If $\iota\edge{a_1}s_1\edge{a_2}\dots\edge{a_m}s$ is the unique path $P_s$  from $\iota$ to $s$ in $A'$, then
\begin{equation}\label{eq:path_1}
sup(s)=sup(\iota)-con(a_1)+pro(a_1) \dots -con(a_m) +pro(a_m)
\end{equation}
If we define $\textbf{z}=(pro(e_1)-con(e_1),\dots, pro(e_n)- con(e_n)) (\in\mathbb{Z}^n)$, then Equation~\ref{eq:path_1} reduces to
\begin{equation}\label{eq:path_2}
sup(s)=sup(\iota) + \psi_s\cdot \textbf{z}
\end{equation}
with the classical scalar product of two vectors in $\mathbb{Z}^n$:
$(x_1,\dots, x_n)\cdot(y_1,\dots, y_n)=x_1\cdot y_1 + \dots + x_n\cdot y_n$.

\medskip
Armed with these notions, we are now able to introduce $L_\alpha$.
For the sake of simplicity, we first restrict ourselves to the case that $\alpha$ equals an (arbitrary but fixed) ESSA  $(e_j,q)$ of $A$, where $j\in \{1,\dots, n\}$.
We will see later that the following arguments also apply to SSA.

\medskip
For every fundamental cycle $\psi_t$ of $A'$, $L_\alpha$ has the following equation:
\begin{equation}\label{eq:fundamental}
\psi_t\cdot (x_{n+1}-x_1,\dots, x_{2n}-x_n)=0
\end{equation}

Each variable must be nonnegative, hence
\begin{align}
\label{eq:sup} 0 \leq \, \underbrace{x_0}_{sup(\iota)}  \land\,\,
\forall i\in\{1,\dots, n\}:\,0 \leq &\, \underbrace{x_i}_{con(e_i)}  \land\, 0 \leq \, \underbrace{x_{n+i}}_{pro(e_i)}
\end{align}

Moreover, for all $s\in S$ and $i\in \{1,\dots, n\}$ with $s\edge{e_i}$ in $A$ (where $e_i$ is not necessarily defined at $s$ in $A'$), $L_\alpha$ has the following inequalities:
\begin{align}
 \label{eq:sig} 0\leq &\, \underbrace{x_0 + \psi_s \cdot (x_{n+1}-x_1,\dots, x_{2n}-x_n)}_{sup(s)}  - \underbrace{x_i}_{con(e_i)}
\end{align}

Finally, for the ESSA $\alpha$, $L_\alpha$  has following inequality:
\begin{align}\label{eq:essp}
\underbrace{x_0 + \psi_q \cdot (x_{n+1}-x_1,\dots, x_{2n}-x_n)}_{sup(q)} - \underbrace{x_j}_{con(e_j)} <0
\end{align}

If $R$ solves $\alpha$, then $\textbf{y}=(sup(\iota), con(e_1),\dots, con(e_n),pro(e_1),\dots, pro(e_n))$ solves $L_\alpha$:
Equation~\ref{eq:fundamental} follows by Proposition~6.16 of~\cite{txtcs/BadouelBD15};
Equation~\ref{eq:sup}, Equation~\ref{eq:sig} and Equation~\ref{eq:essp} follow from Equation~\ref{eq:path_2}, the definition of regions, implying $sup(\iota)\geq 0$ and $con(e_i)\leq sup(s)$ for every state $s$ of $A$ at which $e_i$ occurs, and the fact that $R$ solves $\alpha$, implying $sup(q) < con(e_j)$ and thus $sup(q)-con(e_j)<0$.

Furthermore, if $\textbf{x}=(x_0,x_1,\dots, x_n,x_{n+1},\dots, x_{2n})$ is an integer solution of the system $L_\alpha$, then
$(sup(\iota),con(e_1),\dots, con(e_n), pro(e_1),\dots, pro(e_n))=\textbf{x}$ (implicitly) defines a region of $A$ that solves $\alpha$~\cite{txtcs/BadouelBD15}:
Equation~\ref{eq:fundamental}, ensures that defining $sup(s)$ for all $s\in S$ according to Equation~\ref{eq:path_2} yields a support $sup$ that satisfies $sup(s')=sup(s)-con(e)+pro(e)$ for every edge $s\edge{e}s'$ of $A$ (no incoherence may occur from different paths to the same state).
Equation~\ref{eq:sup} ensures a valid support value for $\iota$.
Equations~\ref{eq:sig} ensures $sup(s)\geq con(e_i)$ if the event $e_i$ occurs at $s$ in $A$ and, by $sup(\iota)\geq 0$, the equations also ensure $sup(s)\geq 0$ for all $s\in S$.
Hence, the result $R=(sup, con, pro)$ is a well-defined region of $A$.
Moreover, Equation~\ref{eq:essp} implies that $R$ solves $\alpha$.

Consequently, an integer vector $\textbf{x}$ solves the system built by the Equations~1-6 if and only if $\textbf{x}=(sup(\iota),con(e_1),\dots, con(e_n), pro(e_1),\dots, pro(e_n))$ for some region $R=(sup, con, pro)$ of $A$
that solves $\alpha$.

\medskip
For an SSP $\alpha=(s,s')$, we simply have to replace constraint~(\ref{eq:essp}) by the following, which expresses that the corresponding supports are different
\begin{align}\label{eq:ssp}
(\psi_s-\psi_{s'}) \cdot (x_{n+1}-x_1,\dots, x_{2n}-x_n)\neq 0
\end{align}

The solution of each separation problem (decision and construction if positive) thus always amounts to solving an
integer linear programming (ILP for short) problem \cite{AMP77-9}, of a size linear in the size of $A$
(without an economic function to optimize, but we may always add a null function for that).
It may also be considered as an instance of a satisfiability modulo theory (SMT for short)~\cite{faia/BarrettSST09}.
There are many ILP- or SMT-solvers\footnote{If a solver does not allow strict comparators, it is always possible to replace
``$<0$'' by ``$\leq -1$'', and ``$\neq 0$'' by two systems, one with ``$\leq -1$'' and one with ``$\geq 1$''}
that may be used to solve such systems.
Note that, in practice, it is not necessary to handle each separation problem from scratch: it is a good idea to
first check if one of the regions found previously does not also solve the separation atom at hand.
 It is also possible to handle the various separation problems in parallel, and if $A$ presents some symmetries it is often possible
 to use them to simplify the searches.

Those problems are generally NP-complete, hence they are always in NP, but some subclasses may be polynomial
and this is the case in our context. Indeed, it may be observed that our systems are homogeneous (no constant arises, but $0$).
Hence, if we have a solution, we shall get another one by multiplying the first one by a strictly positive scalar.
That means that, instead of searching a solution in the integer domain, we may always search one in the rational domain,
and afterwards multiply it by an adequate factor to get an integer solution.
Moreover, since there is a single strict constraint (the one corresponding to the separation property under consideration),
we may always replace the constraint $expr<0$ by $expr\leq -1$ and replace each SSA system by two systems with
$expr\leq -1$ and $expr\geq 1$.
The reason for the replacement is that in the rational domain, there are polynomial algorithms to solve such systems of linear constraint,
like Khachiyan's method~\cite[pp. 168-170]{networks/Rajan90}.

\medskip
Now let's deduce the announced non-deterministic polynomial-time algorithm that decides whether $\alpha$ is solvable by a $(\varrho,\kappa)$-restricted region.
If $R$ is such a region, that is, $\vert\preset{R}\vert \leq \varrho$ 
and
$\vert\postset{R}\vert\leq \kappa$, 
 then there are at most $\varrho$ indices $i_1,\dots, i_\varrho\in \{1,\dots, n\}$ and at most $\kappa$ indices $j_1,\dots, j_\kappa\in \{1,\dots, n\}$, such that $pro(e_{i_\ell})>0$ 
 for all $\ell\in \{1,\dots, \varrho\}$ and $con(e_{j_k})>0$ 
 for all $k\in \{1,\dots, \kappa\}$, respectively.
In particular, for all $i\in \{1,\dots, n\}\setminus\{i_1,\dots, i_\varrho\}$ and all $j\in \{1,\dots, n\}\setminus\{j_1,\dots, j_\kappa\}$,
the searched vector $\textbf{x}$ solves the following equations,
in addition to (\ref{eq:fundamental}) - (\ref{eq:essp} or \ref{eq:ssp}):
\begin{align}
\label{eq:con} x_j &=0 \\
\label{eq:pro}  x_{i+n}&=0
\end{align}
The corresponding system $L'_\alpha$ is still homogeneous, and again it has an integer solution if and only
if it has a rational solution which may be decided and constructed in polynomial time in term of the size of $A$.
The corresponding region $R$ will then be a $(\varrho,\kappa)$-restricted region of $A$ solving the separation problem $\alpha$,
detected and built polynomially.

Obviously, if $R$ exists, then a Turing machine $T$ can guess the indices $i_1,\dots, i_\varrho$ and $j_1,\dots, j_\kappa$
in a non-deterministic computation.
After that, $T$ can deterministically construct $L'_\alpha$ and compute an integer solution $\textbf{x}$, if it exists.
The construction of $L'_\alpha$ (after the aforementioned indices are guessed) and the computation of a solution $\textbf{x}$ of $L_\alpha'$ as well as the verification that $\textbf{x}$ actually defines a sought region are doable in polynomial-time.
Altogether, we have argued, that the solvability of a separation property of $A$ by a $(\varrho,\kappa)$-restricted region can be decided by a non-deterministic Turing-machine in polynomial-time.

\medskip
Let $m=\vert A\vert$ denote the size of the input TS $A=(S,E,\delta,\iota)$, implying $\vert E\vert=n \leq m$.
If $\varrho$ and $\kappa$ are fixed in advance, then there are $\binom{n}{\varrho}$ and $\binom{n}{\kappa}$ possibilities to choose $\{e_{i_1},\dots, e_{i_\varrho}\}$ and $\{e_{j_1},\dots, e_{j_\kappa}\}$, respectively.
Hence, in order to decide whether $\alpha $ is solvable by a properly restricted region, we have to check the solvability of at most $\mathcal{O}(m^{\varrho+\kappa})$ systems of equations and inequalities like the ones discussed above.
For every system, its construction and the test of its solvability can be done deterministically in time polynomial in $m$.
Moreover, every solution of a system implies a solving region and there are at most $\vert S\vert\cdot\vert E\vert+\vert S\vert^2$ separation atoms (and this number obviously depends polynomially on $m$).
Finally, by Lemma~\ref{lem:admissible}, any admissible set $\mathcal{R}$ implies already a sought net $N=N_A^{\mathcal{R}}$.
Therefore, this implies that there is a constant $c$ and an algorithm that (deterministically) solves \ers\ in time $\mathcal{O}(m^{\varrho+\kappa+c})$ and we have the following corollary:

\begin{corollary}\label{cor:fixed}
For any fixed natural numbers $\varrho$ and $\kappa$, there is a constant $c$ and an algorithm that runs in time $\mathcal{O}(m^c)$ and decides whether for a given transition system $A$ there is a Petri net $N$ such that (1) the reachability graph of $N$ is isomorphic to $A$ and (2) every place $p$ of $N$ satisfies $\vert \preset{p}\vert \leq \varrho$ and $\vert \postset{p}\vert \leq \kappa$ and, in the event of a positive decision, constructs a sought net $N$.
\end{corollary}

From the well-known symmetry property  $\binom{n}{k}= \binom{n}{n-k}$, we deduce more generally that
\begin{corollary}\label{cor:fixedgen}
For any fixed natural numbers $\varrho$ and $\kappa$, there is a polynomial algorithm that decides whether for a given transition system $A$ there is a Petri net $N$ such that (1) the reachability graph of $N$ is isomorphic to $A$ and (2) every place $p$ of $N$ satisfies $\vert \preset{p}\vert \leq n_1$ and $\vert \postset{p}\vert \leq n_2$ and, in the event of a positive decision, constructs a sought net $N$, for the pairs $(n_1,n_2)=(\varrho,\kappa)$ or $(|E|-\varrho,\kappa)$ or $(\varrho,|E|-\kappa)$ or
$(|E|-\varrho,|E|-\kappa)$ or $(|E|,\kappa)$ or $(\varrho,|E|)$ or $(|E|,|E|-\kappa)$ or $(|E|-\varrho,|E|)$.
\end{corollary}

This argument is valid if $\varrho$ and $\kappa$ are fixed in advance, but might fail in some cases where  $\varrho$ and $\kappa$ rely on the size of the TS $A$ to be solved.
 For instance, if $\varrho=|E|/2$, there are $\binom{|E|}{|E|/2}$ ways to choose $\{e_{i_1},\dots, e_{i_{|E|/2}}\}$,
 and  $\binom{|E|}{|E|/2}$ grows like $(4^{|E|}) / \sqrt{\pi\cdot|E|/2}$,  hence exponentially
 (it is close to the Catalan number $C_{|E|/2}$).
 This is not a proof that the problem is exponential and not polynomial, since there could be another way
 to tackle the problem, but should make us suspicious. And we shall show in the next subsection that the problem is indeed
 NP-complete when $\varrho$ and $\kappa$ are given together with $A$.

\subsection{Environment Restricted Synthesis is NP-hard}\label{sec:np_hard}%

In order to complete the proof of Theorem~\ref{the:main_result}, it remains to show that \textsc{ERS} is NP-hard.
Our proof of the NP-hardness will be based on a polynomial-time reduction of the hitting set problem,
which is known to be NP-complete (see~\cite{coco/Karp72}):

\noindent
\fbox{\begin{minipage}[t][1.7\height][c]{0.97\textwidth}
\begin{decisionproblem}
  \problemtitle{\textsc{Hitting Set (HS)}}
  \probleminput{A triple $(\mathfrak{U},M,\lambda)$ that consist of a finite set $\mathfrak{U}$,
  a set $M=\{M_0,\dots, M_{m-1}\}$ of subsets of $\mathfrak{U}$ and a natural number $\lambda$.}
  \problemquestion{Does there exist a hitting set $\mathfrak{S}$ for $(\mathfrak{U},M)$, that is, some subset
  $\mathfrak{S}\subseteq \mathfrak{U}$ such that $\mathfrak{S}\cap M_i\not=\emptyset$ for all $i\in \{0,\dots, m-1\}$,
  that satisfies $\vert \mathfrak{S}\vert \leq \lambda$?}
\end{decisionproblem}
\end{minipage}}

\begin{example}\label{ex:hs}
The instance $(\mathfrak{U},M,3)$ such that $\mathfrak{U}=\{X_0,X_1,X_2,X_3\}$ and $M=\{M_0,\dots, M_5\}$, where $M_0=\{X_0,X_1\}$, $M_1=\{X_0,X_2\}$, $M_2=\{X_0,X_3\}$, $M_3=\{X_1,X_2\}$, $M_4=\{X_1,X_3\}$ and $M_5=\{X_2,X_3\}$, allows a positive decision:
$\mathfrak{S}=\{X_0,X_1,X_2\}$ is a suitable hitting set.
\end{example}

In the remainder of this section, until stated explicitly otherwise, let $(\mathfrak{U}, M, \lambda)$ be an arbitrary but fixed input of \textsc{HS} such that  $\mathfrak{U}=\{X_0,\dots, X_{n-1}\}$ and $M=\{M_0,\dots,M_{m-1}\}$, where $M_i=\{X_{i_0},\dots, X_{i_{m_i-1}}\}$ (and thus $\vert M_i\vert =m_i$) for all $i\in \{0,\dots, m-1\}$.
For technical reasons, we assume without loss of generality that $i_0 <\dots < i_{m_i-1}$ for the elements  $X_{i_0},\dots, X_{i_{m_i-1}}$ of the set $M_i$ for all $i\in \{0,\dots, m-1\}$.
Moreover, still for technical reasons, we assume that $\lambda\geq 5$.
Notice that this is not a restriction of the generality, since the hitting set problem is polynomial
for every fixed $\lambda$~\cite{sp/CyganFKLMPPS15}.

\begin{remark}\label{rem:size_lambda}
Obviously, the input of Example~\ref{ex:hs} does not satisfy $\lambda\geq 5$.
However, in order to be able to provide a complete example of the reduction despite the space restrictions, this input is deliberately chosen to be small.
\end{remark}

\textbf{The Reduction.}
In order to prove the hardness part of Theorem~\ref{the:main_result}, we start from input $(\mathfrak{U}, M, \lambda)$ and construct an input $(A,\varrho,\kappa)$ such that the elements of $\mathfrak{U}$ (plus some others) occur as events in the TS $A$.
Moreover, by construction, the TS $A$ has an ESSA $\alpha$ such that the following implication is true:
If $R=(sup, con, pro)$ is a region such that $\vert \preset{R}\vert \leq \varrho$ and $\vert \postset{R}\vert \leq \kappa$ that solves $\alpha$, then the set $\mathfrak{S}=\{X\in \mathfrak{U}\mid pro(X)>0\}$ is a sought HS with at most $\lambda$ elements for $(\mathfrak{U}, M)$.
Consequently, if $(A,\varrho,\kappa)$ allows a positive decision, then there is an admissible set of regions $\mathcal{R}$ whose pre- and postsets are accordingly restricted.
In particular, there is a region $R\in \mathcal{R}$ that solves $\alpha$ and thus proves that $(\mathfrak{U}, M,\lambda)$ also allows a positive decision.
Conversely, we argue that if $(\mathfrak{U}, M, \lambda)$ has a fitting hitting set, then there is an admissible set $\mathcal{R}$ of $A$ such that $\vert \preset{R}\vert \leq \varrho$ and $\vert \postset{R}\vert \leq \kappa$ for all $R\in \mathcal{R}$.
Altogether, this approach proves that $(\mathfrak{U}, M, \lambda)$ is a yes-instance if and only if $(A,\varrho,\kappa)$ is a yes-instance.
Hence, if we were to have a polynomial algorithm to decide \ers, we would also have one for \hs, which is impossible.

\begin{figure}[ht!]
\vspace{2mm}
\begin{center}
\begin{tikzpicture}[new set = import nodes]
\coordinate (top0) at (0,1);
\coordinate (top1) at (3.4,1);
\coordinate (top2) at (6.8,1);
\foreach \i in {top0} {\fill[green!20, rounded corners] (\i) +(-0.4,-0.4) rectangle +(7,0.4);}
\foreach \i in {top0, top1,top2} {\fill[green!20, rounded corners] (\i) +(-0.4,-0.4) rectangle +(0.6,1.6);}
\foreach \i in {top0} {\fill[green!20, rounded corners] (\i) +(-1.4,-11.25) rectangle +(2.5,0.4);}
\begin{scope}[yshift=2.2cm,nodes={set=import nodes}]

	\foreach \i in {0,...,1} {\coordinate (f0\i) at (\i*2.1cm,0);}
	\foreach \i in {0,...,1} {\node (f0\i) at (f0\i) {\nscale{$f_{0,\i}$}};}
	\graph {
	(import nodes);
			f00->[bend left=35, "\escale{$k$}"]f01;
			f00->[swap, bend right=30, "\escale{$k_0$}"]f01;
			f00<-[ "\escale{$z_0$}"]f01;
			
			};
\end{scope}
\begin{scope}[xshift=3.4cm, yshift=2.2cm, nodes={set=import nodes}]

	\foreach \i in {0,...,1} {\coordinate (f1\i) at (\i*2.1cm,0);}
	\foreach \i in {0,...,1} {\node (f1\i) at (f1\i) {\nscale{$f_{1,\i}$}};}
	\graph {
	(import nodes);
			f10->[bend left=35, "\escale{$k$}"]f11;
			f10->[swap, bend right=30, "\escale{$k_1$}"]f11;
			f10<-[ "\escale{$z_1$}"]f11;
			
			};
\end{scope}
\begin{scope}[xshift=6.8cm, yshift=2.2cm,nodes={set=import nodes}]

	\foreach \i in {0,...,1} {\coordinate (f2\i) at (\i*2.1cm,0);}
	\foreach \i in {0,...,1} {\node (f2\i) at (f2\i) {\nscale{$f_{2,\i}$}};}
	\graph {
	(import nodes);
			f20->[bend left=35, "\escale{$k$}"]f21;
			f20->[swap, bend right=30, "\escale{$k_2$}"]f21;
			f20<-[ "\escale{$z_2$}"]f21;
			
			};
\end{scope}
\node (top0) at (0,1) {\nscale{$\top_0$}};
\node (top1) at (3.4,1) {\nscale{$\top_1$}};
\node (top2) at (6.8,1) {\nscale{$\top_2$}};
\path (top0) edge [->] node[left] {\escale{$b_0$} } (f00);
\path (top1) edge [->] node[left] {\escale{$b_1$} } (f10);
\path (top2) edge [->] node[left] {\escale{$b_2$} } (f20);
%
\begin{scope}[nodes={set=import nodes}]
		
		\coordinate (bot0) at (0,0);
		\coordinate (bot1) at (0,-1.2);
		\coordinate (bot2) at (0,-2.4);
		\coordinate (bot3) at (0,-3.6);
		\coordinate (bot4) at (0,-4.8);
		\coordinate (bot5) at (0,-6);
		\foreach \i in {0,...,4} {\coordinate (a\i) at (\i*1.8cm+1.8cm,0);}
		\foreach \i in {0,...,4} {\coordinate (b\i) at (\i*1.8cm+1.8cm,-1.2);}
		\foreach \i in {0,...,4} {\coordinate (c\i) at (\i*1.8cm+1.8cm,-2.4);}
		\foreach \i in {0,...,4} {\coordinate (d\i) at (\i*1.8cm+1.8cm,-3.6);}
		\foreach \i in {0,...,4} {\coordinate (e\i) at (\i*1.8cm+1.8cm,-4.8);}
		\foreach \i in {0,...,4} {\coordinate (f\i) at (\i*1.8cm+1.8cm,-6);}
		\foreach \i in {a2,a4,b2,b4,c2,c3,d2,d4,e2,e3,f2,f3} {\fill[green!20, rounded corners] (\i) +(-0.4,-0.4) rectangle +(0.4,0.4);}
		\foreach \i in {a3,b3,d3} {\fill[blue!20, rounded corners] (\i) +(-0.4,-0.4) rectangle +(0.4,0.4);}
		\node (bot0) at (bot0) {\nscale{$\bot_0$}};
		\node (bot1) at (bot1) {\nscale{$\bot_1$}};
		\node (bot2) at (bot2) {\nscale{$\bot_2$}};
		\node (bot3) at (bot3) {\nscale{$\bot_3$}};
		\node (bot4) at (bot4) {\nscale{$\bot_4$}};
		\node (bot5) at (bot5) {\nscale{$\bot_5$}};
		
		\foreach \i in {0,...,4} {\node (a\i) at (a\i) {\nscale{$t_{0,\i}$}};}
		\foreach \i in {0,...,4} {\node (b\i) at (b\i) {\nscale{$t_{1,\i}$}};}
		\foreach \i in {0,...,4} {\node (c\i) at (c\i) {\nscale{$t_{2,\i}$}};}
		\foreach \i in {0,...,4} {\node (d\i) at (d\i) {\nscale{$t_{3,\i}$}};}
		\foreach \i in {0,...,4} {\node (e\i) at (e\i) {\nscale{$t_{4,\i}$}};}
		\foreach \i in {0,...,4} {\node (f\i) at (f\i) {\nscale{$t_{5,\i}$}};}
		
		\graph {
	(import nodes);
			bot0->["\escale{$y_0$}"]a0->["\escale{$k$}"]a1->["\escale{$X_0$}"]a2->["\escale{$X_1$}"]a3->["\escale{$k$}"]a4;
			bot1->["\escale{$y_1$}"]b0->["\escale{$k$}"]b1->["\escale{$X_0$}"]b2->["\escale{$X_2$}"]b3->["\escale{$k$}"]b4;
			bot2->["\escale{$y_2$}"]c0->["\escale{$k$}"]c1->["\escale{$X_0$}"]c2->["\escale{$X_3$}"]c3->["\escale{$k$}"]c4;
			bot3->["\escale{$y_3$}"]d0->["\escale{$k$}"]d1->["\escale{$X_1$}"]d2->["\escale{$X_2$}"]d3->["\escale{$k$}"]d4;
			bot4->["\escale{$y_4$}"]e0->["\escale{$k$}"]e1->["\escale{$X_1$}"]e2->["\escale{$X_3$}"]e3->["\escale{$k$}"]e4;
			bot5->["\escale{$y_5$}"]f0->["\escale{$k$}"]f1->["\escale{$X_2$}"]f2->["\escale{$X_3$}"]f3->["\escale{$k$}"]f4;
			};
\end{scope}
\path (bot0) edge [->] node[left] {\nscale{$a_0$} } (top0);
\path (top0) edge [->] node[above] {\nscale{$a_1$} } (top1);
\path (top1) edge [->] node[above] {\nscale{$a_2$} } (top2);

\path (bot0) edge [->] node[left] {\nscale{$w_1$} } (bot1);
\path (bot1) edge [->] node[left] {\nscale{$w_2$} } (bot2);
\path (bot2) edge [->] node[left] {\nscale{$w_3$} } (bot3);
\path (bot3) edge [->] node[left] {\nscale{$w_4$} } (bot4);
\path (bot4) edge [->] node[left] {\nscale{$w_5$} } (bot5);
\path (a1) edge [->] node[left] {\nscale{$u_1$} } (b1);
\path (b1) edge [->] node[left] {\nscale{$u_2$} } (c1);
\path (c1) edge [->] node[left] {\nscale{$u_3$} } (d1);
\path (d1) edge [->] node[left] {\nscale{$u_4$} } (e1);
\path (e1) edge [->] node[left] {\nscale{$u_5$} } (f1);
%
\foreach \i in {1,...,5} {\coordinate (tr\i) at (\i*2cm-2cm,-7cm);}
\foreach \i in {tr1} {\fill[green!20, rounded corners] (\i) +(-1.4,-3.25) rectangle +(8.7,0.4);}
\begin{scope}[yshift=-8cm,nodes={set=import nodes}]

	\foreach \i in {0,...,1} {\coordinate (g1\i) at (0,-\i*2cm);}
	\foreach \i in {0,...,1} {\node (g1\i) at (g1\i) {\nscale{$g_{1,\i}$}};}
	\graph {
	(import nodes);
			g10<-[bend left=20, "\escale{$u_1$}"]g11;
			g10->[swap, bend right=20, "\escale{$v_1$}"]g11;
			
			};
\end{scope}
\begin{scope}[xshift=2cm, yshift=-8cm,nodes={set=import nodes}]

	\foreach \i in {0,...,1} {\coordinate (g2\i) at (0,-\i*2cm);}
	\foreach \i in {0,...,1} {\node (g2\i) at (g2\i) {\nscale{$g_{2,\i}$}};}
	\graph {
	(import nodes);
			g20<-[bend left=20, "\escale{$u_2$}"]g21;
			g20->[swap, bend right=20, "\escale{$v_2$}"]g21;
			
			};
\end{scope}
\begin{scope}[xshift=4cm, yshift=-8cm, nodes={set=import nodes}]

	\foreach \i in {0,...,1} {\coordinate (g3\i) at (0,-\i*2cm);}
	\foreach \i in {0,...,1} {\node (g3\i) at (g3\i) {\nscale{$g_{3,\i}$}};}
	\graph {
	(import nodes);
			g30<-[bend left=20, "\escale{$u_3$}"]g31;
			g30->[swap, bend right=20, "\escale{$v_3$}"]g31;
			
			};
\end{scope}
\begin{scope}[xshift=6cm,yshift=-8cm,nodes={set=import nodes}]

	\foreach \i in {0,...,1} {\coordinate (g4\i) at (0,-\i*2cm);}
	\foreach \i in {0,...,1} {\node (g4\i) at (g4\i) {\nscale{$g_{4,\i}$}};}
	\graph {
	(import nodes);
			g40<-[bend left=20, "\escale{$u_4$}"]g41;
			g40->[swap, bend right=20, "\escale{$v_4$}"]g41;
			
			};
\end{scope}
\begin{scope}[xshift=8cm, yshift=-8cm,nodes={set=import nodes}]

	\foreach \i in {0,...,1} {\coordinate (g5\i) at (0,-\i*2cm);}
	\foreach \i in {0,...,1} {\node (g5\i) at (g5\i) {\nscale{$g_{5,\i}$}};}
	\graph {
	(import nodes);
			g50<-[bend left=20, "\escale{$u_5$}"]g51;
			g50->[swap, bend right=20, "\escale{$v_5$}"]g51;
			
			};
\end{scope}
\foreach \i in {1,...,5} {\node (tr\i) at (tr\i) {\nscale{$\triangle_\i$}};}

\path (tr1) edge [->] node[left] {\nscale{$d_1$} } (g10);
\path (tr2) edge [->] node[left] {\nscale{$d_2$} } (g20);
\path (tr3) edge [->] node[left] {\nscale{$d_3$} } (g30);
\path (tr4) edge [->] node[left] {\nscale{$d_4$} } (g40);
\path (tr5) edge [->] node[left] {\nscale{$d_5$} } (g50);
\coordinate (bot0_1) at (-0.75,0);
\coordinate (tr1_1) at (-0.75,-7);
\path (bot0) edge [-] node[left] {} (bot0_1);
\path (bot0_1) edge [-] node[left] {\nscale{$c_1$} } (tr1_1);
\path (tr1_1) edge [->] node[left] {} (tr1);
\path (tr1) edge [->] node[above] {\nscale{$c_2$} } (tr2);
\path (tr2) edge [->] node[above] {\nscale{$c_3$} } (tr3);
\path (tr3) edge [->] node[above] {\nscale{$c_4$} } (tr4);
\path (tr4) edge [->] node[above] {\nscale{$c_5$} } (tr5);
\end{tikzpicture}
\end{center}\vspace*{-2mm}
\caption{The TS $A$ with initial state $\bot_0$ that results from Example~\ref{ex:hs}.
Based on the $3$-HS $\{X_0,X_1,X_2\}$ for $(\mathfrak{U}, M)$, the colored area sketches
the region $R_1$ of Fact~\ref{fact:k} that solves $(k,t_{0,1})$:
the support of the states in the green colored area equals $1$, states in the blue colored area have support $2$,
and the other ones have support $0$.}\label{fig:reduction_example}
\end{figure}
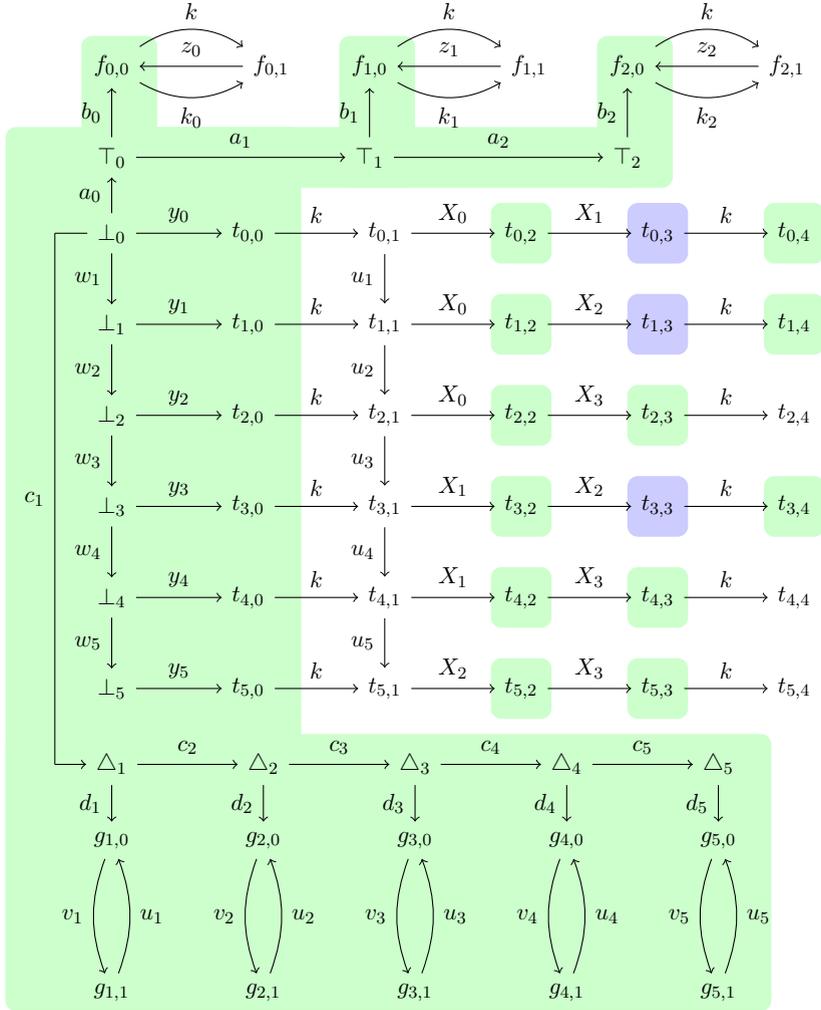

\medskip
First of all, we define $\varrho=2\cdot\lambda$ and $\kappa=\lambda+1$.
Notice that this implies $\varrho\geq 10$ and $\kappa\geq 6$, since we assume $\lambda\geq 5$.
However, that does not restrict the generality, since \ers{} is polynomial for any fixed integers.
Figure~\ref{fig:reduction_example} provides a complete example for the following construction, which is based on the input of Example~\ref{ex:hs}.
The TS $A$ has, for every $i\in \{0,\dots, m-1\}$, the following gadget $T_i$ that represents the set $M_i=\{X_{i_0}, \dots, X_{i_{m_i-1}}\}$ by using its elements as events (between two occurrences of the event $k$):\vspace*{-2mm}
\begin{center}
\begin{tikzpicture}[new set = import nodes]
\begin{scope}[nodes={set=import nodes}]
	
	\node (t) at (-0.75,0){$T_i=$};
	\foreach \i in {0,...,2} {\coordinate (\i) at (\i*2cm,0);}
	\foreach \i in {3} {\coordinate (\i) at (\i*2cm+6,0);}
	\foreach \i in {4} {\coordinate (\i) at (8.5,0);}
	\foreach \i in {0,1} {\node (\i) at (\i) {\nscale{$t_{i,\i}$}};}
	\node (2) at (2) {$\dots$};
	\node (3) at (3) {$t_{i,m_i+1}$};
	\node (4) at (4) {$t_{i,m_i+2}$};
	\graph {
	(import nodes);
			0->["\escale{$k$}"]1->["\escale{$X_{i_0}$}"]2->["\escale{$X_{i_{m_i-1}}$}"]3->["\escale{$k$}"]4;
			
			};
\end{scope}
\end{tikzpicture}
\end{center}
In particular, $T_0$ provides the ESSA $\alpha=(k,t_{0,1})$.
Additionally, the states $t_{i,1}$ and $t_{i+1,1}$ are connected by an $u_{i+1}$-labeled edge $t_{i,1}\edge{u_{i+1}}t_{i+1,1}$ for all $i\in \{0,\dots, m-2\}$.
Moreover, for every $i\in \{0,\dots, \lambda-1\}$, the TS $A$ has the gadget $F_i$ and, for every $j\in \{1,\dots, m-1\}$, it has the gadget $G_j$ that uses the event $u_j$ again:
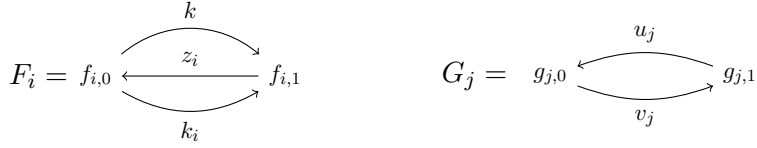
\vspace*{-2mm}
\begin{center}
\begin{tikzpicture}[new set = import nodes]
\begin{scope}[nodes={set=import nodes}]
	
	\node (t) at (-0.75,0){$F_i=$};
	\foreach \i in {0,...,1} {\coordinate (a\i) at (\i*2.5cm,0);}
	\foreach \i in {0,...,1} {\node (a\i) at (a\i) {\nscale{$f_{i,\i}$}};}
	\graph {
	(import nodes);
			a0->[bend left=40, "\escale{$k$}"]a1;
			a0->[swap, bend right=30, "\escale{$k_i$}"]a1;
			a0<-["\escale{$z_i$}"]a1;
			
			};
\end{scope}
\begin{scope}[xshift=6cm, nodes={set=import nodes}]
	
	\node (t) at (-1,0){$G_j=$};
	\foreach \i in {0,...,1} {\coordinate (a\i) at (\i*2.5cm,0);}
	\foreach \i in {0,...,1} {\node (a\i) at (a\i) {\nscale{$g_{j,\i}$}};}
	\graph {
	(import nodes);
			a0<-[bend left=20, "\escale{$u_j$}"]a1;
			a0->[swap, bend right=20, "\escale{$v_j$}"]a1;
			
			};
\end{scope}
\end{tikzpicture}
\end{center}

The functional part of $A$ is given by the introduced gadgets.
The initial state of $A$ is $\bot_0$.
In order to connect the gadgets, we use the following edges that introduce fresh events and states:
For every $i\in \{0,\dots, m-1\}$, the TS $A$  has the edges $\bot_{i}\edge{y_i}t_{i,0}$ and, if $i < m-1$, the edge $\bot_{i}\edge{w_{i+1}}\bot_{i+1}$;
the TS $A$ has the edge $\bot_0\edge{a_0}\top_0$ and, for every $i\in \{0,\dots, \lambda-1\} $, it has the edge $\top_i\edge{b_i}f_{i,0}$ and, if $i < \lambda-1$, it has the edge $\top_{i}\edge{a_{i+1}}\top_{i+1}$;
the TS $A$ has the edge $\bot_0\edge{c_1}\triangle_1$ and, for every $i\in \{1,\dots, m-1\}$, it has the edge $\triangle_i\edge{d_i}g_{i,0}$ and if $i < m-1$, then it has the edge $\triangle_i\edge{c_{i+1}}\triangle_{i+1}$.
By $S$ and $E$ we refer to the (set of) states and (set of) events of $A$, respectively.

\begin{lemma}\label{lem:essp_implies_vc}
If there is an admissible set of $(\varrho,\kappa)$-restricted regions for $A$, then there is a hitting set with at most $\lambda$ elements for $(\mathfrak{U}, M)$.
\end{lemma}
\begin{proof}
Let $\mathcal{R}$ be an admissible set of $A$ that satisfies $\vert \preset{R}\vert\leq\varrho$ and $\vert \postset{R}\vert\leq \kappa$ for all $R\in \mathcal{R}$.
Since $\mathcal{R}$ is admissible, there is an accordingly restricted region that solves the atom $\alpha=(k,t_{0,1})$.
Let $R=(sup, con, pro)$ be such a region, that is $con(k) > sup(t_{0,1})$, $\vert \preset{R}\vert \leq 2\cdot\lambda$
and $\vert \postset{R}\vert \leq \lambda +1$.
In the following, we argue that $\mathfrak{S}=\{X\in \mathfrak{U}\mid pro(X)> 0\}$ defines a sought hitting set for $(\mathfrak{U}, M)$.

\medskip
Since $k$ occurs at $t_{0,0}$ and $R$ solves $\alpha$, the following is true:\\
(1) $con(k) \leq sup(t_{0,0})$\\
(2) $sup(t_{0,1}) =sup(t_{0,0})-con(k)+pro(k)$\\
(3) $con(k) > sup(t_{0,1})$. \medskip\\
By combining (1) and (2), we obtain $sup(t_{0,1})\geq pro(k)$.
If we combine the latter with (3), then we get $con(k) > pro(k)$.
For all $i\in \{0,\dots, \lambda-1\}$, by $con(k) > pro(k)$ and $f_{i,0}\edge{k}f_{i,1}$,
we conclude $sup(f_{i,0}) > sup(f_{i,1})$, since $sup(f_{i,1})=sup(f_{i,0})-con(k)+pro(k)$.
This implies $con(k_i) > pro(k_i)$ as well as $con(z_i) < pro(z_i)$;
thus $k_i\in \postset{R}$ and $z_i\in \preset{R}$ for all $i\in \{0,\dots,\lambda-1\}$.
Since $k\in \postset{R}$, this implies already $\vert \postset{R}\vert =\lambda +1$.
In particular, no further event can be a member of $\postset{R}$.
Let $i\in \{1,\dots, m-1\}$ be arbitrary but fixed.
If $pro(u_i) > con(u_i)$, then we obtain $sup(g_{i,0})> sup(g_{i,1})$ and thus $con(v_i) > pro(v_i)$ by $g_{i,1}\edge{u_i}g_{i,0}$ and $g_{i,0}\edge{v_i}g_{i,1}$.
This would imply $v_i\in \postset{R}$ and $\vert \postset{R}\vert \geq \lambda+2$, a contradiction.
Hence, we have $pro(u_i) \leq con(u_i)$.
By $t_{i-1,1}\edge{u_i}t_{i,1}$, this implies $sup(t_{i-1,1})\geq sup(t_{i,1})$.
Since $i$ was arbitrary, we get $sup(t_{0,1})\geq sup(t_{1,1})\geq \dots\geq sup(t_{m-1,1})$.
By $con(k)>sup(t_{0,1})$, this implies $con(k)>sup(t_{i,1})$ for all $i\in \{0,\dots, m-1\}$.
Furthermore, since $t_{i,m_i+1}\edge{k}$, we have $con(k)\leq sup(t_{i,m_i+1})$ for all $i\in \{0,\dots, m-1\}$.
Consequently, there has to be at least one event $X\in \{X_{i_0}, \dots, X_{i_{m_i-1}}\}$ such that $pro(X)> 0$.
This implies $\mathfrak{S}\cap M_i\not=\emptyset$ for all $i\in \{0,\dots, m-1\}$, where $\mathfrak{S}=\{X\in \mathfrak{U}\mid pro(X)> 0\}$.
Moreover, since $z_i\in \preset{R}$ for all $i\in \{0,\dots,\lambda-1\}$ and $\mathfrak{S}\subseteq \preset{R}$ and $\vert \preset{R}\vert \leq 2\cdot\lambda$, we have that $\vert \mathfrak{S}\vert\leq \lambda$.
This proves the claim and thus the lemma.
\end{proof}

In order to show that all ESSA are solvable by $(\varrho,\kappa)$-restricted regions, provided there is a fitting hitting set for $(\mathfrak{U}, M)$, we treat the events of $A$ individually.
Recall that an event is solvable if all its corresponding ESSA are solvable (Definition~\ref{def:essp}).
Moreover, for a region $R=(sup, con, pro)$ of $A$, the set $\T_{c,p}^{R}$ contains the events $e\in E$ such that $(con(e), pro(e))=(c,p)$ (Remark~\ref{rem:implicitly}).

\begin{fact}\label{fact:k}
If there is a hitting set with at most $\lambda$ elements for $(\mathfrak{U}, M)$, then the event $k$ is solvable by $(\varrho,\kappa)$-restricted regions.
\end{fact}
\begin{proof}
Let $\mathfrak{S}$ be a hitting set with at most $\lambda$ elements for $(\mathfrak{U}, M)$.

The following region $R_0=(sup_0, con_0, pro_0)$ solves $(k,s)$ for all $s\in \bigcup_{i=1}^{m-1}S(G_i)$ and all $s\in \{\top_0,\dots, \top_{\lambda-1}\}$ and all $s\in \{\triangle_1,\dots, \triangle_{m-1}\}$:
$sup_0(\bot_0)=1$ and
$\T_{1,1}^{R_0}=\{k\}$ and 
$\T_{0,1}^{R_0}=\{b_0,\dots, b_{\lambda-1}\}$ and 
$\T_{1,0}^{R_0}=\{a_0,c_1\}$ and 
$\T_{0,0}^{R_0}=E\setminus (\T_{1,1}^{R_1} \cup \T_{0,1}^{R_1}\cup \T_{1,0}^{R_1})$. 
This region satisfies $\vert \preset{R_0}\vert =\lambda+1 \leq 2\cdot\lambda$ and $\vert \postset{R_0}\vert=3$.

The following region $R_1=(sup_1, con_1, pro_1)$ solves $\alpha=(k,t_{0,1})$ and, moreover, $(k,s)$ for all $s\in \{f_{i,1}\mid i\in \{0,\dots, \lambda-1\}\}$:
$sup_1(\bot_0)=1$;
$\T_{1,0}^{R_1}=\{k,k_0,\dots, k_{\lambda-1}\}$ and 
$\T_{0,1}^{R_1}= \mathfrak{S} \cup\{z_0,\dots, z_{\lambda-1}\}$ and 
$\T_{0,0}^{R_1}= E\setminus  (\T_{1,0}^{R_1}\cup \T_{0,1}^{R_1})$. 
This region satisfies $\vert \preset{R_1}\vert \leq 2\cdot\lambda$, since $\vert \mathfrak{S} \vert\leq\lambda$, and $\vert \postset{R_1}\vert=\lambda+1$.

The following region $R_2=(sup_2, con_2, pro_2)$ solves $\alpha=(k,s)$ for all $s\in \{t_{i,1}, t_{i,m_i+2}\mid i\in \{0,\dots, m-1\}\}$:
$sup_2(\bot_0)=2$ and
$\T_{1,0}^{R_2}=\{k,k_0,\dots, k_{\lambda-1}\}$ and
$\T_{0,1}^{R_2}=\{z_0,\dots, z_{\lambda-1}\}$ and
$\T_{0,0}^{R_2}=E\setminus (\T_{1,0}^{R_2}\cup \T_{0,1}^{R_2})$.
This region satisfies $\vert \preset{R_1}\vert =\lambda \leq 2\cdot\lambda$ and $\vert \postset{R_1}\vert=\lambda+1$.

The next region $R_3=(sup_3, con_3, pro_3)$ solves $(k,s)$ for $s=\bot_0$ and all $s\in \{t_{0,2},\dots, t_{0,m_0}\}$:
$sup_3(\bot_0)=0$ and
$\T_{1,1}^{R_3}=\{k\}$ and 
$\T_{0,1}^{R_3}=\{y_0, u_1, X_{0_{m_0-1}},a_0,c_1\}\}$ and 
$\T_{0,2}^{R_3}=\{w_1\}$ and 
$\T_{1,0}^{R_3}= \{X_{0_0},v_1\}$ and 
$\T_{0,0}^{R_3}=E\setminus (\T_{1,1}^{R_3}\cup \T_{0,1}^{R_3} \cup \T_{0,2}^{R_3}\cup \T_{1,0}^{R_3})$. 

Let $i\in \{1,\dots, m-1\}$ be arbitrary but fixed.
The following region $R_4=(sup_4, con_4, pro_4)$ solves $(k,s)$ for all $s\in \{\bot_i, t_{i,2},\dots, t_{i,m_i}\}$:
$sup_4(\bot_0)=2$ and 
$\T_{1,1}^{R_4}=\{k\}$  and 
$\T_{2,0}^{R_4}=\{w_i\}$ and 
$\T_{0,2}^{R_4}=\{w_{i+1}\}$ and  
$\T_{0,1}^{R_4}= \{y_i,v_i,u_{i+1}, X_{i_{m_i-1}}\}$ and 
$\T_{1,0}^{R_4}=  \{u_i,v_{i+1},X_{i_0}\}$ and 
$\T_{0,0}^{R_4}= E\setminus (\T_{1,1}^{R_4}\cup \T_{2,0}^{R_4}\cup \T_{0,2}^{R_4}\cup \T_{0,1}^{R_4}\cup \T_{1,0}^{R_4})$. 
By the arbitrariness of $i$, this proves the solvability of $k$.
\end{proof}

\begin{fact}\label{fact:u}
If $e\in \{u_1,\dots, u_{m-1}\}$, then $e$ is solvable by $(\varrho,\kappa)$-restricted regions.
\end{fact}
\begin{proof}
Let $i\in \{1,\dots, m-1\}$ be arbitrary but fixed.
The following region $R_5=(sup_5, con_5, pro_5)$ solves $(u_i, s)$ for all states $s\in S\setminus (S(T_{i-1})\cup\{g_{i-1,0}\})$ with $s\edge{\neg u_i}$:
If $i=1$, then $sup_5(\bot_0)=1$, otherwise $sup(\bot_0)=0$; 
if $i=1$, then $\T_{1,0}^{R_5}= \{u_i, w_i,v_{i-1}\}\cup \{a_0,c_1\}$, else $\T_{1,0}^{R_5}= \{u_i, w_i,v_{i-1}\}$; and 
$\T_{0,1}^{R_5}=\{w_{i-1},u_{i-1},d_{i-1}, v_i\}$ and 
$\T_{0,0}^{R_5}= E\setminus (\T_{1,0}^{R_5}\cup \T_{0,1}^{R_5})$. 

Region $R_6=(sup_6, con_6, pro_6)$ solves $(u_i, s)$ for all $s\in \{\bot_{i-1},t_{i-1,0}\}$ and if $i\geq 2$, then for $s\in\{g_{i-1,0}\}$:
$sup_6(\bot_0)=0$ and 
$\T_{1,1}^{R_6}=\{u_i\}$ and 
$\T_{0,1}^{R_6}= \{d_i, k,k_0,\dots,k_{\lambda-1}\}$ and 
$\T_{1,0}^{R_6}= \{z_0,\dots,z_{\lambda-1}\} $ and 
$\T_{0,0}^{R_6} = E\setminus (\T_{1,1}^{R_6}\cup \T_{0,1}^{R_6}\cup \T_{1,0}^{R_6})$. 
Notice that $\vert \preset{R_6}\vert =\lambda+3\leq 2\cdot\lambda$, since $\lambda\geq 6$, and $\vert \postset{R_6}\vert =\lambda+1$.

Finally, the region $R_7=(sup_7, con_7, pro_7)$ solves $(u_i, s)$ for all $s\in \{t_{i-1,2}, \dots, t_{i-1,m_i+2}\}$:
$sup_7(\bot_0)=1$ and 
$\T_{1,1}^{R_7}= \{u_i\}$ and 
$\T_{1,0}^{R_7}= \{X_{i_0}\}$ and 
$\T_{0,0}^{R_7}=E\setminus (\T_{1,1}^{R_7}\cup \T_{1,0}^{R_7})$. 
This completes solving $u_i$ and, by the arbitrariness of $i$, this proves the solvability for all $e\in \{u_1,\dots, u_{m-1}\}$.
\end{proof}

\begin{fact}\label{fact:X}
If $e\in \{X_0,\dots, X_{n-1}\}$, then $e$ is solvable by $(\varrho,\kappa)$-restricted regions.
\end{fact}
\begin{proof}
Let $i\in \{0,\dots, n-1\}$ be arbitrary but fixed.
Moreover, let $j,\ell\in \{0,\dots, m-1\}$ be arbitrary but fixed such that $X_i\not\in M_j$ and $X_i\in M_\ell$.

The next region $R_8=(sup_8, con_8, pro_8)$ solves $(X_i, s)$ for $s= \bot_j$ and all $s\in\{t_{j,0}, \dots, t_{j,m_j+2}\}$:
If $j=0$, then $sup_8(\bot_0)=0$, otherwise $sup(\bot_0)=1$; 
$\T_{1,1}^{R_8}=\{X_i\}$ and 
if $j > 0$, then $\T_{1,0}^{R_8}=\{w_j,u_j,v_{j+1}\}\cup \{a_0,c_1\}$, else $\T_{1,0}^{R_8}=\{w_j,u_j,v_{j+1}\}$;
$\T_{0,1}^{R_8}= \{v_j,w_{j+1},u_{j+1}, d_{j+1}\}$ and 
$\T_{0,0}^{R_8}= E\setminus (\T_{1,1}^{R_8}\cup \T_{0,1}^{R_8})$. 

The next region $R_9=(sup_9, con_9, pro_9)$ solves $(X_i, s)$ for all $s\in \{\bot_\ell,t_{\ell,0}\}$:
$sup_9(\bot_0)=0$; 
$\T_{1,1}^{R_9}= \{X_i\}$ and 
$\T_{1,0}^{R_9}= \{z_0,\dots, z_{\lambda-1}\}$ and 
$\T_{0,1}^{R_9}= \{k,k_0,\dots, k_{\lambda-1}\}$ and 
$\T_{0,1}^{R_9}=E\setminus (\T_{1,1}^{R_9}\cup \T_{1,0}^{R_9}\cup \T_{0,1}^{R_9})$. 
Notice that $\vert \postset{R_9}\vert =\lambda+1$ and $\vert \preset{R_9}\vert =\lambda+1\leq 2\cdot\lambda$.

Let $h\in \{0,\dots, m_\ell-1\}$ be the unique index such that $X_i=X_{\ell_h}$, that is, $X_i$ is the ``$h$-th element'' of $M_\ell$.
The following region $R_{10}=(sup_{10}, con_{10}, pro_{10})$ solves $(X_i, s)$ for all $s\in \{t_{\ell,h+1}, t_{\ell, m_\ell+2}\}$:
$sup_{10}(\bot_0)=1$ and 
$\T_{1,0}^{R_{10}}=\{X_i,a_0,c_1\}$ and 
$\T_{0,0}^{R_{10}}= E\setminus \T_{1,0}^{R_{10}}$. 

It remains to discuss the case $X_i\not=X_{\ell_0}$, that is $h\geq 1$, which requires to solve $(X_i,s)$ for all $s\in \{t_{\ell,1},\dots, t_{\ell,h}\}$.
So let $s\in \{t_{\ell,1},\dots, t_{\ell,h}\}$ be arbitrary but fixed.

We distinguish between $\ell=0$ and $\ell \geq 1$:
If $\ell=0$, then the region $R_{11}=(sup_{11}, con_{11}, pro_{11})$ solves $(X_i, s)$:
$sup_{11}(\bot_0)=0$ and 
$\T_{1,0}^{R_{11}}= \{X_i,v_1\}$ and 
$\T_{0,1}^{R_{11}}= \{w_1, u_1, d_1,X_{0_{h-1}}\}$ and 
$\T_{0,0}^{R_{11}}= E\setminus (\T_{1,0}^{R_{11}}\cup \T_{0,1}^{R_{11}})$. 

If $\ell\geq 1$, then region $R_{12}=(sup_{12}, con_{12}, pro_{12})$ solves $(X_i, s)$:
$sup_{12}(\bot_0)=1$ and 
$\T_{1,0}^{R_{12}}= \{X_i,w_\ell,u_\ell,v_{\ell+1},a_0,c_1\}$ and 
$\T_{0,1}^{R_{12}}= \{X_{i_{h-1}},w_{\ell+1}, u_{\ell+1}, v_\ell, d_{\ell+1}\}$ and  
$\T_{0,0}^{R_{12}}= E\setminus (\T_{1,0}^{R_{12}}\cup \T_{0,1}^{R_{12}})$. 
By the arbitrariness of $h$, this completes the solvability of $(X_i,s)$ for all $s\in S(T_\ell)$.

The following region $R_{13}=(sup_{13}, con_{13}, pro_{13})$ solves $(X_i, s)$ for all $s\in S\setminus (\bigcup_{j=0}^{m-1}(S(T_j)\cup\{\bot_j\})$:
$sup_{13}(\bot_0)=1$ and 
$\T_{1,1}^{R_{13}}=\{X_i\}$ and 
$\T_{1,0}^{R_{13}}=\{a_0,c_1\}$  and 
$\T_{0,0}^{R_{13}}=E\setminus (\T_{1,1}^{R_{13}}\cup \T_{1,0}^{R_{13}})$. 
Since $i$, $j$, $\ell$ were arbitrary, we have the claim.
\end{proof}


\begin{fact}\label{fact:essp_for_secondary_events}
If $e\in \{k_0,\dots, k_{\lambda-1}\}$ or $e\in \{z_0,\dots, z_{\lambda-1}\}$ or $e\in \{v_1,\dots, v_{m-1}\}$ or $e\in\{w_1,\dots, w_{m-1}\}$ or $e\in \{a_0,\dots, a_{\lambda-1}\}$ or $e\in\{b_0,\dots, b_{\lambda-1}\}$ or $e\in\{y_0,y_i,w_i,c_i,d_i\mid 1 \leq i\leq m-1\}$, then $e$ is solvable by $(\varrho,\kappa)$-restricted regions.
\end{fact}
\begin{proof}
Let $i\in \{0,\dots, \lambda-1\}$ be arbitrary but fixed.
The region $R_1$ of Fact~\ref{fact:k} solves $(k_i, f_{i,1})$ and the region $R_6$ of Fact~\ref{fact:u} solves $(z_i,f_{i,0})$.
It is easy to see that $(k_i,s)$ and $(z_i,s)$ are suitably solvable for the remaining $s\in S\setminus\{f_{i,0},f_{i,1}\}$.
Since $i$ was arbitrary, that proves the claim for all $e\in \{k_0,\dots, k_{\lambda-1}, z_0,\dots, z_{\lambda-1}\}$.

Let $i\in \{1,\dots, m-1\}$ be arbitrary but fixed.
The region $R_6$ or the region $R_8$ solves $(v_i,g_{i,1})$.
It is easy to see that $(v_i,s)$ is suitably solvable for all $s\in S\setminus\{g_{i,0}, g_{i,1}\}$.
By their uniqueness, it is easy to see that the remaining events are also solvable by suitably restricted regions.
The claim follows.
\end{proof}

Altogether, the just presented facts prove that all ESSA of $A$ are solvable by $(\varrho,\kappa)$-restricted regions, if there is a hitting set with at most $\lambda$ elements for $(\mathfrak{U},M)$.
Moreover, if $(s,s')$ is an SSA of $A$, then either $(s,s')$ is already solved by one of the presented regions or it is easy to see that a solving $(\varrho,\kappa)$-restricted region exists.
Hence, we obtain the following lemma, which completes the proof of Theorem~\ref{the:main_result}.

\begin{lemma}\label{lem:vc_implies_essp}
If there is a hitting set with at most $\lambda$ elements for $(\mathfrak{U},M)$, then $A$ has an admissible set $\mathcal{R}$ of $(\varrho,\kappa)$-restricted regions.
\end{lemma}


\subsection{A lower bound for the parameterized complexity of \textsc{ERS}}\label{sec:para}%

By Theorem~\ref{the:main_result}, the problem \textsc{ERS} is NP-complete.
Hence, from the point of view of the classical complexity theory, where we assume that P is different from NP,
the problem is considered intractable, i.e., the worst-case time-complexity of any deterministic decision algorithm is
above polynomial (as it is up to now).
However, measuring the complexity of the problem purely in terms of the size of the input may let it appear harder than it actually is.

In the \emph{parameterized complexity} theory, we deal with \emph{parameterized problems},
where every input $(x,k)$ has a distinguished part $k$ (a natural number), called the parameter,
and measure the complexity not only in terms of the input size $n$, but also in terms of the parameter $k$.

For example, the natural parameter of \ers\ is $k=\varrho+\kappa$.
From the results of Section~\ref{sec:in_np}, there is an algorithm that solves \ers\ in time $\mathcal{O}(n^{k+c})$,
where $c$ is a constant.
In terms of parameterized complexity, this means that \ers\ parameterized by $k$ belongs to the complexity class XP (for \emph{slice-wise polynomial}).
However, such algorithms are not considered as feasible, since $n^{k+c}$ can be huge even for small $k$ and moderate $n$.
Hence, we are rather interested in algorithms where  $k$ does not appear in the exponent of $n$:
We say a parameterized problem is \emph{fixed parameter tractable} if it has an algorithm with running time $\mathcal{O}(f(k)n^c)$, where $f$ is a computable function that depends only on $k$, and $c$ is a constant.
Such algorithms are manageable even for large values of $n$, provided that $f(k)$ is relatively small and $c$ is a small constant.

In order to obtain a successful parameterization, we need to have some reason to believe that the parameter is typically small, such that $f(k)$ can be expected to remain relatively small, too.
In the absence of a benchmark that is specifically created for Petri net synthesis, we have analyzed the benchmark of the Model Checking Contest (MCC)~\cite{tacas/AmparoreBCDGHHJ19}, which contains both academic and industrial Petri nets.
Their corresponding TS (reachability graphs) have usually (way) more than $10^7$ states, which provides a lower bound for the length $n\geq \vert S\vert +\vert E\vert$ of an input TS $A=(S,E,\delta,\iota)$.
We analyzed 878 Petri nets in total.
For 395 of them (around 45\%), we found that $\varrho+\kappa\leq 21$, and for 308 of them (around 35\%) we even found $\varrho+\kappa\leq 11$ (for example, \texttt{AutoFlight-PT-02a} and \texttt{CircadianClock-PT-100000} and \texttt{FMS-PT-50000}).
In other words: depending on $f$, a synthesis algorithm with running time $f(k)n^c$ could possibly be useful for a third up to almost half of these nets.
From this point of view, the parameterization of \ers\ by $\varrho+\kappa$ and the search for a fixed-parameter algorithm appear to be sensible.
Unfortunately, we can provide strong evidence that such an algorithm does not exists.
This is of practical relevance, since it prevents an algorithm designer to waste countless hours with the attempt to find a solution that most likely cannot be found.

From the classical complexity theory  viewpoint, a problem is considered as intractable if it is NP-hard.
Analogously, in parameterized complexity, a problem is assumed not fixed-parameter-tractable
if it is $W[i]$-hard for some $i\geq 1$.
We omit the formal definition of the complexity class $W[i]$ and rather refer to~\cite{sp/CyganFKLMPPS15}.
In order to show that a parameterized problem $Q$ is $W[i]$-hard, we have to present a parameterized reduction from a known $W[i]$-hard problem $P$ to $Q$.
A \emph{parameterized reduction} is an algorithm that transforms an instance $(x,k)$ of $P$ into an instance $(x',k')$ of $Q$ such that
\begin{enumerate}
\item $(x,k)$ is in $P$ if and only if $(x',k')$ is in $Q$ and
\item $k'\leq g(k)$ for some computable function independent of $x$, and
\item its running time is $f(k)\vert x\vert^c$ for some computable function $f$ and constant $c$.
\end{enumerate}

The \textsc{Hitting Set} problem parameterized by $\lambda$ is known to be $W[2]$-hard (even $W[2]$-complete).
Moreover, the reduction presented in Section~\ref{sec:np_hard} is a parameterized one, since $\varrho+\kappa=3\lambda+1$.
This proves the following theorem, implying the fixed-parameter-intractability of \ers\ parameterized by $\varrho+\kappa$:
\begin{theorem}
\ers\ parameterized by $\varrho+\kappa$ is $W[2]$-hard.
\end{theorem}

\section{The pure case}\label{sec:new_content}%

In this section, we investigate the computational complexity of the following variant of our original problem:\smallskip

\noindent
\fbox{\begin{minipage}[t][1.7\height][c]{0.97\textwidth}
\begin{decisionproblem}
  \problemtitle{\textsc{Pure Environment Restricted Synthesis}}
  \probleminput{A TS $A=(S,E,\delta, \iota)$ and two natural numbers $\varrho$ and $\kappa$.}
  \problemquestion{Does  there exist an admissible set $\mathcal{R}$ of pure regions of $A$ such that every region $R\in \mathcal{R}$ satisfies $\vert \preset{R}\vert\leq \varrho$ and $\vert \postset{R}\vert \leq \kappa$?}
\end{decisionproblem}
\end{minipage}}
\vspace{0.5cm}

\noindent
In order to express that a region is pure, we could use the constraints
\begin{align}
\label{eq:pure}
\forall e\in E: con(e)\cdot pro(e)=0
\end{align}

But these constraints are quadratic, and not linear.
Following~\cite{txtcs/BadouelBD15}, it is possible to keep constraints linear by using another kind of regions,
composed of only two components $R=(sup,\eff)$ where $sup:S\rightarrow \mathbb{N}$ is the usual support
and $\eff:E\rightarrow \mathbb{Z}$ is the \emph{\underline{eff}ect} $pro(e)-con(e)$ of executing event $e$.
For such a region $R$, $\preset{R}=\{e|\eff(e)>0\}$ and $\postset{R}=\{e|\eff(e)<0\}$,
enforcing the pureness of the corresponding place in a Petri net.
The systems characterizing the regions of a pure synthesis of some TS $A$ may then be adapted and
remain composed of polynomially many  linear constraints.
However, we shall not need this new kind of regions here.

\begin{lemma}\label{lem:PERS-NP}
Pure Environment Restricted Synthesis is in NP.
\end{lemma}

\begin{proof}
For each separation atom $\alpha$ (there are polynomially many of them), a Turing machine $T$ can guess two subsets
$E_p$ and $E_c$ of $E$ with $|E_p|\leq\varrho$, $|E_c|\leq\kappa$ and $E_c\cap E_p=\emptyset$,
and we may search for a region $R$ corresponding to the system $L_\alpha$ (as introduced in Section~\ref{sec:in_np}) extended with the constraints (on \textbf{x}) corresponding to:
\[\forall e\not\in E_p: pro(e)=0, \text{ and } \\
\forall e\not\in E_c: con(e)=0\]

This system is composed of polynomially many homogeneous linear constraints;
hence we may again search for a solution in the rational domain instead of the integer one, before renormalizing the solution into the integer domain, and this may be decided and computed polynomially.
\end{proof}

If $\varrho$ and $\kappa$ are fixed beforehand, there are polynomial algorithms to choose $E_c$ and $E_p$, so that we have

\begin{corollary}\label{cor:pure-fixed}
For any fixed natural numbers $\varrho$ and $\kappa$, there is a polynomial algorithm that decides whether for a given transition system $A$ there is a pure Petri net $N$ such that
(1) the reachability graph of $N$ is isomorphic to $A$, and
(2) every place $p$ of $N$ satisfies $\vert \preset{p}\vert \leq \varrho$, and $\vert \postset{p}\vert \leq \kappa$ and,
in the event of a positive decision, constructs a sought net $N$.
\end{corollary}

It remains to prove the NP-hardness of the \textsc{Pure Environment Restricted Synthesis}
when $\varrho$ and $\kappa$ are not fixed beforehand, to get

\begin{theorem}\label{the:new_content}
{Pure Environment Restricted Synthesis} is NP-complete.
\end{theorem}

The proof of Theorem~\ref{the:new_content} is based on a reduction of the following problem:

\noindent
\fbox{\begin{minipage}[t][1.8\height][c]{0.97\textwidth}
\begin{decisionproblem}
  \problemtitle{\textsc{Cubic Monotone 1 in 3 3Sat (CM1in33Sat)}}
  \probleminput{A pair $(\mathfrak{U}, M)$ that consists of a set $\mathfrak{U}$ of   boolean variables and a set of
  $m$ 3-clauses $M=\{M_0,\dots, M_{m-1}\}$ such that $M_i=\{X_{i_0}, X_{i_1}, X_{i_2}\}\subseteq \mathfrak{U}$ and $i_0 <i_1 <i_2$ for all $\in \{0,\dots, m-1\}$.
 Every variable of $\mathfrak{U}$ occurs in exactly three clauses of $M$ }
  \problemquestion{Does there exist a one-in-three model of $(\mathfrak{U}, M)$, i. e., a subset $\mS\subseteq \mathfrak{U}$ such that $\vert \mS\cap M_i\vert =1$ for all $i\in \{0,\dots, m-1\}$?}
\end{decisionproblem}
\end{minipage}}
\medskip

\begin{theorem}[\cite{dcg/MooreR01}]
\textsc{Cubic Monotone 1 in 3 3Sat} is NP-complete.
\end{theorem}

\begin{example}\label{ex:cubic}
	\label{ex:13sat}
	The instance $ (\mathfrak{U},M) $, where $ \mathfrak{U} =\{X_0,X_1,X_2,X_3,X_4,X_5\} $, and $ M =\{M_0,\dots, M_5\}$ such that $M_0=\{X_0,X_1,X_2\}$,
$M_1=\{X_0,X_1,X_3\}$,
$M_2=\{X_0,X_2,X_3\}$, 
$M_3=\{X_1,X_4,X_5\}$, 
$M_4=\{X_2,X_4,X_5\}$, and
$M_5=\{X_3,X_4,X_5\}$, allows a positive decision: $\mS=\{X_0,X_4\}$ defines a one-in-three model for $ (\mathfrak{U},M) $.
\end{example}

In the following, until explicitly stated otherwise, let $(\mathfrak{U}, M)$ be an arbitrary but fixed instance of \textsc{CM1in33Sat} with variables $\mathfrak{U}=\{X_0,\dots, X_{m-1}\}$, and clauses $M=\{M_0,\dots, M_{m-1}\}$, where $M_i=\{X_{i_0}, X_{i_1}, X_{i_2}\}\subseteq \mathfrak{U}$, and $i_0 <i_1 <i_2$ for all $i \in \{0,\dots, m-1\}$.
Note that $\vert \mathfrak{U}\vert =\vert M\vert$ holds by the definition of a valid input.

\begin{lemma}\label{lem:model_size}
$\mS\subseteq \U$ is a one-in-three model of $(\U, M)$ if and only if $\mS\cap M_i\not=\emptyset$ for all $i\in \{0,\dots, m-1\}$, and $\vert \mS\vert =\frac{m}{3}$.
\end{lemma}
\begin{proof}
Every variable of $\U$ occurs in exactly three distinct clauses.
Hence, every set $\mS\subseteq \U$ intersects with $3\vert\mS\vert$ (distinct) clauses $M_{i_0},\dots, M_{i_{3\vert\mS\vert-1}}\in M$ if and only if $\vert \mS\cap M_{i_j}\vert=1 $ is satisfied for all $j\in \{0,\dots, 3\vert\mS\vert-1\}$.
\end{proof}

We reduce the instance $(\U,M)$ to an instance $(A,\varrho,\kappa)$ with TS $A=(S,E,\delta, h_0)$ as follows:
First, we define $\varrho=m$, and $\kappa=\vert E\vert$.
We proceed by developing stepwise the TS $A$.
First of all, the TS $A$ has the following gadget $H$ 
that uses an event $k$, and,
for all $i\in \{0,\dots, \frac{2m}{3}-1\}$ an event $u_i$, which labels an edge that is converse to edge labeled by $k$ (by Lemma~\ref{lem:model_size} we can assume without loss of generality that $m\equiv 0\text{ mod }3$):\vspace*{-3mm}

\begin{center}
\begin{tikzpicture}[new set = import nodes]
\begin{scope}[nodes={set=import nodes}]
	\node (H) at (-1,0) {$H=$};
	\node (h0) at (0,0) {\nscale{$h_0$}};
	\node (h1) at (2.5,0) {\nscale{$h_1$}};
	\graph {
	(import nodes);
			h0->[bend left =20, "\escale{$k$}"]h1;
			h1->[bend left=20, "\escale{$u_0,\dots, u_{\frac{2m}{3}-1}$}"]h0;
			};
\end{scope}
\end{tikzpicture}\vspace*{-2mm}
\end{center}

Moreover, for every $i\in \{0,\dots, m-1\}$, the TS $A$ has the following gadget $T_i$, which uses the variables of the clause $M_i=\{X_{i_0},X_{i_1}, X_{i_2}\}$ as events:
\begin{center}
\begin{tikzpicture}[new set = import nodes]
\begin{scope}[nodes={set=import nodes}]
	\node (Ti) at (-1,0) {$T_i=$};
	\foreach \i in {0,...,5} {\coordinate (t0\i) at (\i*1.75cm,0);}
	\foreach \i in {0,...,5} {\node (t0\i) at (t0\i) {\nscale{$t_{i,\i}$}};}
	\graph {
	(import nodes);
			t00->[ "\escale{$X_{i_0}$}"]t01->[ "\escale{$X_{i_1}$}"]t02->[ "\escale{$X_{i_2}$}"]t03->["\escale{$v_i^0,\dots, v_i^{\frac{m}{3}-1}$}"]t04->[bend left=20, "\escale{$k$}"]t05;
			t05->[bend left=30, "\escale{$u_0,\dots, u_{\frac{2m}{3}-1}$}"]t04;
			};
\end{scope}
\end{tikzpicture}
\end{center}
Notice that $T_i$ uses the same events $u_0,\dots, u_{\frac{2m}{3}-1}$ as $H$.
On the other hand, for every $i\not=j\in \{0,\dots, m-1\}$, the events $v_i^0,\dots, v_i^{\frac{m}{3}-1}$, and the events $v_j^0,\dots, v_j^{\frac{m}{3}-1}$ are pairwise distinct.

\medskip
Finally, for every $i\in \{0,\dots, m-1\}$, and for every $j\in \{0,\dots, \frac{m}{3}-1\}$, the TS $A$ has the edge $h_1\Edge{w_i^j}t_{i,0}$.

As an illustration, Figure~\ref{fig:pure_reduction_example} shows the TS $A$ that originates from the input of Example~\ref{ex:cubic}.

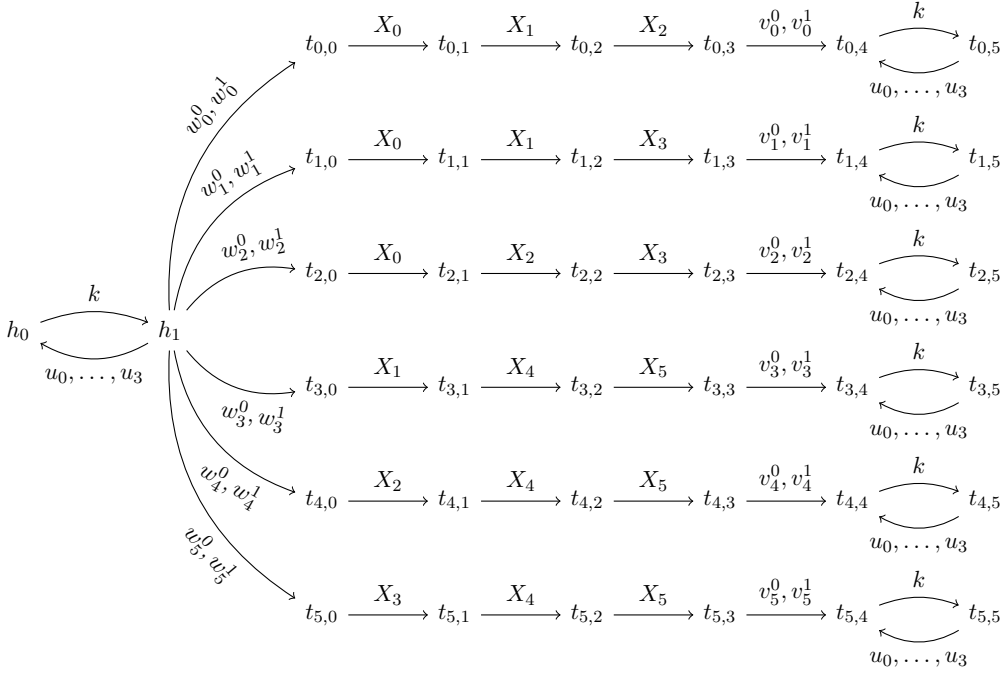
\begin{figure}[t!]
\begin{center}
\begin{tikzpicture}[new set = import nodes]
\begin{scope}[nodes={set=import nodes}]

	\foreach \i in {0,...,5} {\coordinate (t0\i) at (\i*1.75cm,0);}
	\foreach \i in {0,...,5} {\node (t0\i) at (t0\i) {\nscale{$t_{0,\i}$}};}
	\graph {
	(import nodes);
			t00->[ "\escale{$X_0$}"]t01->[ "\escale{$X_1$}"]t02->[ "\escale{$X_2$}"]t03->["\escale{$v_0^0,v_0^1$}"]t04->[bend left=20, "\escale{$k$}"]t05;
			t05->[bend left=30, "\escale{$u_0,\dots, u_{3}$}"]t04;
			};
\end{scope}
\begin{scope}[yshift=-1.5cm, nodes={set=import nodes}]

	\foreach \i in {0,...,5} {\coordinate (t1\i) at (\i*1.75cm,0);}
	\foreach \i in {0,...,5} {\node (t1\i) at (t1\i) {\nscale{$t_{1,\i}$}};}
	\graph {
	(import nodes);
			t10->[ "\escale{$X_0$}"]t11->[ "\escale{$X_1$}"]t12->[ "\escale{$X_3$}"]t13->["\escale{$v_1^0,v_1^1$}"]t14->[bend left=20, "\escale{$k$}"]t15;
			t15->[bend left=30, "\escale{$u_0,\dots, u_3$}"]t14;
			};
\end{scope}
\begin{scope}[yshift=-3cm, nodes={set=import nodes}]

	\foreach \i in {0,...,5} {\coordinate (t2\i) at (\i*1.75cm,0);}
	\foreach \i in {0,...,5} {\node (t2\i) at (t2\i) {\nscale{$t_{2,\i}$}};}
	\graph {
	(import nodes);
			t20->[ "\escale{$X_0$}"]t21->[ "\escale{$X_2$}"]t22->[ "\escale{$X_3$}"]t23->["\escale{$v_2^0,v_2^1$}"]t24->[bend left=20, "\escale{$k$}"]t25;
			t25->[bend left=30, "\escale{$u_0,\dots, u_3$}"]t24;
			};
\end{scope}
\begin{scope}[yshift=-4.5cm, nodes={set=import nodes}]

	\foreach \i in {0,...,5} {\coordinate (t3\i) at (\i*1.75cm,0);}
	\foreach \i in {0,...,5} {\node (t3\i) at (t3\i) {\nscale{$t_{3,\i}$}};}
	\graph {
	(import nodes);
			t30->[ "\escale{$X_1$}"]t31->[ "\escale{$X_4$}"]t32->[ "\escale{$X_5$}"]t33->["\escale{$v_3^0,v_3^1$}"]t34->[bend left=20, "\escale{$k$}"]t35;
			t35->[bend left=30, "\escale{$u_0,\dots, u_3$}"]t34;
			};
\end{scope}
\begin{scope}[yshift=-6cm, nodes={set=import nodes}]

	\foreach \i in {0,...,5} {\coordinate (t4\i) at (\i*1.75cm,0);}
	\foreach \i in {0,...,5} {\node (t4\i) at (t4\i) {\nscale{$t_{4,\i}$}};}
	\graph {
	(import nodes);
			t40->[ "\escale{$X_2$}"]t41->[ "\escale{$X_4$}"]t42->[ "\escale{$X_5$}"]t43->["\escale{$v_4^0,v_4^1$}"]t44->[bend left=20, "\escale{$k$}"]t45;
			t45->[bend left=30, "\escale{$u_0,\dots, u_3$}"]t44;
			};
\end{scope}
\begin{scope}[yshift=-7.5cm, nodes={set=import nodes}]

	\foreach \i in {0,...,5} {\coordinate (t5\i) at (\i*1.75cm,0);}
	\foreach \i in {0,...,5} {\node (t5\i) at (t5\i) {\nscale{$t_{5,\i}$}};}
	\graph {
	(import nodes);
			t50->[ "\escale{$X_3$}"]t51->[ "\escale{$X_4$}"]t52->[ "\escale{$X_5$}"]t53->["\escale{$v_5^0,v_5^1$}"]t54->[bend left=20, "\escale{$k$}"]t55;
			t55->[bend left=30, "\escale{$u_0,\dots, u_3$}"]t54;
			};
\end{scope}
\begin{scope}[yshift=-3.75cm, nodes={set=import nodes}]

	\node (h0) at (-4,0) {\nscale{$h_0$}};
	\node (h1) at (-2,0) {\nscale{$h_1$}};
	\path (h1) edge [->, bend left=30] node[above, rotate=50, pos=0.7] {\escale{$w_0^0,w_0^1$}} (t00);
	\path (h1) edge [->, bend left=30] node[above, rotate=30, pos=0.7] {\escale{$w_1^0,w_1^1$}} (t10);
	\path (h1) edge [->, bend left=30] node[above, rotate=10, pos=0.7] {\escale{$w_2^0,w_2^1$}} (t20);
	\path (h1) edge [->, bend right=30] node[below, rotate=-10, pos=0.7] {\escale{$w_3^0,w_3^1$}} (t30);
	\path (h1) edge [->, bend right=30] node[below, rotate=-30, pos=0.7] {\escale{$w_4^0,w_4^1$}} (t40);
	\path (h1) edge [->, bend right=30] node[below, rotate=-50, pos=0.7] {\escale{$w_5^0,w_5^1$}} (t50);
	\graph {
	(import nodes);
			h0->[bend left =20, "\escale{$k$}"]h1;
			h1->[bend left=30, "\escale{$u_0,\dots, u_3$}"]h0;
			};
\end{scope}
\end{tikzpicture}
\end{center}\vspace*{-4mm}
\caption{The TS $A$ that is the result of the reduction for \textsc{Pure Environment Restricted Synthesis} applied to the input of Example~\ref{ex:cubic}.}\label{fig:pure_reduction_example}\vspace*{-2mm}
\end{figure}

\begin{lemma}\label{lem:pure_essp_implies_model}
If there is an admissible set of pure and $(\varrho,\kappa)$-restricted regions of $A$, then there is a one-in-three model for $(\U,M)$.
\end{lemma}
\begin{proof}
Let $\R$ be an admissible set of pure and $(\varrho,\kappa)$-restricted regions of $A$. Let $R=(sup, con, pro)\in \R$  that solves $\alpha=(k,h_1)$.
By $h_0\edge{k}$, we have $con(k)\leq sup(h_0)$, and since $R$ solves $\alpha$, we have that $con(k)>sup(h_1)$.
This implies $sup(h_0)>sup(h_1)$, as well as $con(k)>0= pro(k)$ (by pureness).
From $h_1\edge{u_i}h_0$, we get $pro(u_i) > con(u_i)=0$, implying $u_i\in \preset{R}$ for all $i\in \{0,\dots, \frac{2m}{3}-1\}$.
Since $\vert \preset{R}\vert \leq \varrho=m$, there are at most $\frac{m}{3}$ events left that have a positive $pro$-value.
In the following, we argue that $\mS=\{X\in \U\mid pro(X)>0\}$ satisfies $\vert \mS\vert= \frac{m}{3}$ and $\mS\cap M_i\not=\emptyset$ for all $i\in \{0,\dots, m-1\}$.
By Lemma~\ref{lem:model_size}, this implies that $\mS$ is a one-in-three model for $(\U, M)$.

Let $i\not=j\in \{0,\dots, m-1\}$ be arbitrary but fixed.
Since $con(k)> sup(h_1)$, $t_{i,4}\edge{k}$ and $t_{j,4}\edge{k}$, there is an event $e$ on every path from $h_1$ to $t_{i,4}$ such that $pro(e)>0= con(e)$, and the same is true for every path from $h_1$ to $t_{j,4}$.
If there is an $\ell\in \{0,\dots, \frac{m}{3}-1\}$ such that $pro(w_i^\ell)> 0=con(w_i^\ell)$, implying $sup(t_{i,0})>sup(h_1)$, then $pro(w_i^{\ell'})=pro(w_i^\ell)$ for each $\ell'\in \{0,\dots, \frac{m}{3}-1\}$, so that $w_i^\ell\in \preset{R}$ for all $\ell\in \{0,\dots, \frac{m}{3}-1\}$.
Since $\vert \preset{R}\vert\leq m$ and $\preset{R}\supseteq \{u_0,\dots,u_{\frac{2m}{3}-1}\}$,
this would imply that there is no event $e$ with $pro(e)> 0$ on any path from $h_1$ to $t_{j,4}$,
 contradicting what we saw before.
Similarly, one argues that there is no $\ell\in \{0,\dots, \frac{m}{3}-1\}$ such that $pro(v_i^\ell)> 0$.
Hence, by the arbitrariness of $i$ and $j$, we obtain that $\preset{R}\cap (\bigcup_{i=0}^{m-1}\{w_i^0,v_i^0,\dots, w_i^{\frac{m}{3}-1},v_i^{\frac{m}{3}-1}\})=\emptyset$.

Consequently, for every $i\in \{0,\dots, m-1\}$, there is an event $X\in \U$ on the path from $t_{i,0}$ to $t_{i,4}$,
such that $pro(X) > 0$.
This implies $\mS\cap M_i\not=\emptyset$ for all $i\in \{0,\dots, m-1\}$.
Moreover, since $\mS$ intersects with $m$ distinct clauses and $\vert \mS\vert \leq \frac{m}{3}$ by $\preset{R}\supseteq \mS$,
we have that $\vert\mS \vert=\frac{m}{3}$.
By Lemma~\ref{lem:model_size}, this implies that $\mS$ defines a one-in-three model of $(\U,M)$.
\end{proof}

Conversely, the following facts show that a one-in-three model for $(\U,M)$ implies the existence of an
admissible set of $(\varrho,\kappa)$-restricted pure regions of $A$.
We abridge $U=\{u_0,\dots, u_{\frac{2m}{3}-1 }\}$, $W_i=\{w_i^0,\dots, w_i^{\frac{m}{3}-1}\}$ and
$V_i=\{v_i^0,\dots, v_i^{\frac{m}{3}-1}\}$ for all $i\in \{0,\dots, m-1\}$, and $W=\bigcup_{i=0}^{m-1}W_i$,
$V=\bigcup_{i=0}^{m-1}V_i$.

\begin{fact}\label{fact:pure_k}
If there is a one-in-three model for $(\U,M)$, then the event $k$ is solvable by $(\varrho,\kappa)$-restricted pure regions.
\end{fact}
\begin{proof}
Let $\mS$ be a one-in-three model of $(\U,M)$.\\
The following region $R_0=(sup_0, con_0, pro_0)$ solves $(k,h_1)$: 
$sup_0(h_0)=1$, and
$\T_{1,0}^{R_0}=\{k\}$, and
$\T_{0,1}^{R_0}=\mS\cup U$.
Notice that $\vert \preset{R_0}\vert=m$, by Lemma~\ref{lem:model_size}.

Let $i\in \{0,\dots, m-1\}$ be arbitrary but fixed.\\
The following region $R_1=(sup_1, con_1, pro_1)$ solves $(k,s)$ for all $s\in \{t_{i,0},\dots, t_{i,3}, t_{i,5}\}$:
$sup_1(h_0)=2$, and
$\T_{1,0}^{R_1}=\{k\}\cup W_i$, and
$\T_{0,1}^{R_1}=U\cup V_i$.
Notice that $\vert \preset{R}\vert=m$.\\
Since $i$ was arbitrary, this completes the proof.
\end{proof}

\begin{fact}\label{fact:pure_u_and_w}
For every $e\in U\cup W$, the event $e$ is solvable by $(\varrho,\kappa)$-restricted pure regions.
\end{fact}
\begin{proof}
Let $u\in U$, and $w\in W$ be arbitrary but fixed.
The following region $R_0=(sup_0, con_0, pro_0)$ solves $(u,s)$ for all $s\in S$ with $s\edge{\neg u}$, and $(w,s)$ for all $s\in S\setminus (\bigcup_{i=0}^{m-1}\{t_{i,5}\})$ with $s\edge{\neg w}$:
$sup_0(h_0)=0$, and
$\T_{1,0}^{R_0}=U\cup W$, and
$\T_{0,1}^{R_0}=\{k\}$.

The following region $R_1=(sup_1, con_1, pro_1)$ solves $(w,s)$ for all $s\in \bigcup_{i=0}^{m-1}\{t_{i,5}\}$:
$sup_1(h_0)=1$, and
$\T_{1,0}^{R_1}= W$.
Since $u$, and $w$ were arbitrary, the fact follows.
\end{proof}

\begin{fact}\label{fact:pure_X}
For every $X\in \U$, the event $X$ is solvable by $(\varrho,\kappa)$-restricted pure regions.
\end{fact}
\begin{proof}
Let $i\in \{0,\dots, m-1\}$ be arbitrary but fixed, and let $j_0,j_1,j_2\in \{0,\dots, m-1\}$ be the pairwise distinct indices, such that $X_i\in M_{j_0}\cap M_{j_1}\cap M_{j_2}$.

The following region $R_0=(sup_0,con_0,pro_0)$ solves $(X_i, s)$ for all $s\in S\setminus (S(T_{j_0})\cup S(T_{j_1})\cup S(T_{j_2}))$, and for all $s\in (S(T_{j_0})\cup S(T_{j_1})\cup S(T_{j_2}))$ such that $X_i$ occurs \emph{before} $s$ on the unique path from $t_{i_\ell,0}$ to $s$, where $\ell\in \{0,1,2\}$:
$sup_0(h_0)=0$, and
$\T_{1,0}^{R_0}=\{X_i\}$, and
$\T_{0,1}^{R_0}=W_{j_0}\cup W_{j_1}\cup W_{j_2}$.
Notice that $\vert \preset{R_0}\vert =\vert W_{j_0}\vert+\vert W_{j_1}\vert +\vert W_{j_2}\vert =m$.
Moreover, for all $\ell\in \{0,1,2\}$, if $X_i=X_{{j_\ell}_0}$, then $(X_i,s)$ is solved for all $s\in S(T_{j_\ell})$.
Hence, it only remains to consider the case that $X_i$ is at the second or third position.

Let $i\in \{0,\dots, m-1\}$, and $j\in \{1,2\}$ be arbitrary but fixed.
Moreover, let $\ell_0\not=\ell_1\in \{0,\dots,m-1\}\setminus \{i\}$ such that $X_{i_j}\in M_{\ell_0}\cap M_{\ell_1}$, that is, $\ell_0$, and $\ell_1$ select the other two occurrences of $X_{i_j}$.
The following region $R_1=(sup_1,con_1,pro_1)$ solves $(X_{i_j}, s)$ for all $s\in \{t_{i,0},\dots, t_{i,j-1}\}$:
$sup_1(h_0)=0$, and
$\T_{1,0}^{R_1}=\{X_{i_j}\}$, and
$\T_{0,1}^{R_1}=\{X_{i_{j-1}}\}\cup W_{\ell_0}\cup W_{\ell_1}$.
Notice that $\vert \preset{R_1}\vert =\frac{2m}{3}+1\leq \varrho$.\\
By the arbitrariness of $i$ and $j$, this completes the proof.
\end{proof}

\begin{fact}\label{fact:pure_v}
For every $v\in V$, the event $v$ is solvable by $(\varrho,\kappa)$-restricted pure regions.
\end{fact}
\begin{proof}
Let $i\in \{0,\dots, m-1\}$, and $v\in V_i$ be arbitrary but fixed.\\
The following region $R_0=(sup_0,con_0,pro_0)$ solves $(v,s)$ for all $s\in \{t_{i,0},t_{i,1}, t_{i,2}, t_{i,4}, t_{i,5}\}$:
$sup(h_0)=0$, and
$\T_{1,0}^{R_0}=V_i$, and
$\T_{0,1}^{R_0}=\{X_{i_2}\}$.

The following region $R_1=(sup_1,con_1,pro_1)$ solves $(v,s)$ for all $s\in S\setminus S(T_i)$:
$sup(h_0)=0$, and
$\T_{1,0}^{R_0}=V_i$, and
$\T_{0,1}^{R_0}=W_i$.
Since $i$, and $v$ were arbitrary, this proves the fact.
\end{proof}

\begin{fact}\label{fact:pure_ssp}
There is a witness of pure, and $(\varrho,\kappa)$-restricted regions for the SSP of $A$.
\end{fact}
\begin{proof}
We argue for $h_0$:\\
Let $i\in \{0,\dots, m-1\}$ be arbitrary but fixed.\\
The region $R_1$ of Fact~\ref{fact:pure_k} solves $(h_0,s)$ for all $s\in \{h_1\}\cup S(T_i)$.
Since $i$ was arbitrary, the solvability of $h_0$ follows.\\

We proceed with $h_1$:\\
- The region $R_1$ of Fact~\ref{fact:pure_k} solves $(h_1,s)$ for all $s\in \bigcup_{i=0}^{m-1}\{t_{i,5}\}$.\\
- The region $R_0$ of Fact~\ref{fact:pure_u_and_w}  solves $(h_1,s)$ for all $s\in S\setminus(\bigcup_{i=0}^{m-1}\{t_{i,5}\})$.\\

We proceed with the states of the $T_i$'s:
Let $i\in \{0,\dots, m-1\}$ be arbitrary but fixed. \\
The following region $R_0=(sup_0,con_0,pro_0)$ solves $(p,q)$ for all $p\in S(T_i)$, and all $q\in S\setminus S(T_i)$:
$sup_0(h_0)=0$, and
$\T_{0,1}^{R_0}=W_i$.

The following region $R_1=(sup_1,con_1,pro_1)$ solves $(p,q)$ for all $p\not=q\in \{t_{i,0},\dots, t_{i,4}\}$:
$sup_1(h_0)=0$, and $\T_{0,1}^{R_1}=\{X_{i_0},X_{i_1}, X_{i_2}\}\cup V_i$.
Notice that $\vert \preset{R_1}\vert =\frac{m}{3}+3\leq m$.

The following region $R_2=(sup_2,con_2,pro_2)$ solves $(s,t_{i,5})$ for all $s\in \{t_{i,0},\dots, t_{i,4}\}$:
$sup_2(h_0)=2$, and
$\T_{1,0}^{R_2}=\{k\}$.
$\T_{0,1}^{R_2}=U$.
Notice that $\vert \preset{R_1}\vert =\frac{2m}{3}\leq m$. 

Since $i$ was arbitrary, this completes the proof.
\end{proof}

Altogether, by proving Fact~\ref{fact:pure_k} to Fact~\ref{fact:pure_ssp}, we have shown that the following lemma, which closes the proof of Theorem~\ref{the:new_content}, is true:

\begin{lemma}\label{lem:pure_model_implies_essp}
If there is a one-in-three model for $(\U,M)$, then there is an admissible set of pure and $(\varrho,\kappa)$-restricted regions of $A$.
\end{lemma}

\section{Conclusion}\label{sec:conclusion}%

In this paper, we investigated the computational complexity of synthesizing Petri nets for which the cardinality of the pre- and postset of their places is restricted by natural numbers $\varrho$ and $\kappa$.
We showed that the problem is solvable in polynomial time for any fixed $\varrho$ and $\kappa$.
By contrast, if $\varrho$ and $\kappa$ are part of the input, then the resulting \textsc{ERS} problem is NP-complete.
Moreover, we show that \textsc{ERS} parameterized by $\varrho+\kappa$ is $W[2]$-hard and thus most likely does not allow a fixed-parameter-tractable-algorithm.
We then extended our results to the class of pure Petri nets.

Future work could focus on other classes as synthesis target, and on the implementation of our algorithms.
On the other hand, instead of considering other classes, one could also focus on the actual main problem implicitly addressed in this paper, that is, the underlying optimization problem that searches, for a TS $A$, for a realizing Petri net with places whose presets and postsets are as small as possible.
Our results imply that such an optimum solution can not be found efficiently (unless P=NP).
However, instead of an exact solution, a good approximate one may be sufficient, that is, we may investigate whether the corresponding optimization problem allows a so-called $c$-approximation algorithm for a constant $c\geq 1$.

\subsubsection*{Acknowledgements.}%

We would like to thank Karsten Wolf, who provided a summary of the data from the Model Checking Contest.
Moreover, we would like to thank the anonymous reviewers of the original version of this paper~\cite{apn/Tredup21a}, and this extended version for their detailed comments and valuable suggestions.


\end{document}